\documentclass[oribibl, 11pt]{llncs}
\usepackage{fullpage}

\usepackage{amsmath}
\usepackage{amssymb}
\usepackage{graphicx}
\usepackage{subfigure}
\usepackage{algorithm}
\usepackage{algpseudocode}
\usepackage{color}
\usepackage{hyperref}
\usepackage{gensymb}

\usepackage{fancyhdr}
\makeatletter
\def\blfootnote{\xdef\@thefnmark{}\@footnotetext}
\makeatother

\renewcommand{\algorithmicrequire}{\textbf{Initialization:}}

\newcommand{\removed}[1]{}

\newcommand{\rem}{\mathbf}

\newcommand{\cu}{\mathcal{U}}

\newcommand{\ca}{\mathcal{A}}

\newcommand{\cc}{\mathcal{C}}

\newcommand{\cl}{\mathcal{L}}

\newcommand{\ra}{\rightarrow}
\newcommand{\rsa}{\rightsquigarrow}

\newcommand{\coleq}{\mathrel{\mathop:}=}
\newcommand{\bs}{\backslash}

\newcommand{\dprime}{{\prime\prime}}

\renewcommand{\P}{\textup{P}}

\definecolor{DarkRed}{RGB}{182,11,1}

\title{Terminating Distributed Construction of Shapes and Patterns in a Fair Solution of Automata\thanks{Supported in part by the project ``Foundations of Dynamic Distributed Computing Systems'' (\textsf{FOCUS}) which is implemented under the ``ARISTEIA'' Action of the  Operational Programme ``Education and Lifelong Learning'' and is co-funded by the European Union (European Social Fund) and Greek National Resources.}}

\author{Othon Michail}
\institute{Computer Technology Institute \& Press ``Diophantus'' (CTI), Patras, Greece\\
Email:\email{ michailo@cti.gr}\\ Phone: +30 2610 960300}

\begin{document}

\maketitle

\begin{abstract}
In this work, we consider a \emph{solution of automata} similar to \emph{Population Protocols} and \emph{Network Constructors}. The automata (also called \emph{nodes}) move passively in a well-mixed solution without being capable of controlling their movement. However, the nodes can \emph{cooperate} by interacting in pairs. Every such interaction may result in an update of the local states of the nodes. Additionally, the nodes may also choose to connect to each other in order to start forming some required structure. We may think of such nodes as the \emph{smallest possible programmable pieces of matter}, like tiny nanorobots or programmable molecules. The model that we introduce here is a more applied version of Network Constructors, imposing \emph{physical} (or \emph{geometrical}) \emph{constraints} on the connections that the nodes are allowed to form. Each node can connect to other nodes only via a very limited number of \emph{local ports}, which implies that at any given time it has only a \emph{bounded number of neighbors}. Connections are always made at \emph{unit distance} and are \emph{perpendicular to connections of neighboring ports}. Though such a model cannot form abstract networks like Network Constructors, it is still capable of forming very practical \emph{2D or 3D shapes}. We provide direct constructors for some basic shape construction problems, like \emph{spanning line}, \emph{spanning square}, and \emph{self-replication}. We then develop \emph{new techniques} for determining the computational and constructive capabilities of our model. One of the main novelties of our approach, concerns our attempt to overcome the inability of such systems to detect termination. In particular, we exploit the assumptions that the system is well-mixed and has a unique leader, in order to \emph{give terminating protocols that are correct with high probability}. This allows us to develop terminating subroutines that can be \emph{sequentially composed} to form larger \emph{modular protocols} (which has not been the case in the relevant literature). One of our main results is a \emph{terminating protocol counting the size $n$ of the system} with high probability. We then use this protocol as a subroutine in order to develop our \emph{universal constructors}, establishing that \emph{it is possible for the nodes to become self-organized with high probability into arbitrarily complex shapes while still detecting termination of the construction}.
\end{abstract}


\section{Introduction}
\label{sec:intro}

Recent research in distributed computing theory and practice is taking its first timid steps on the pioneering endeavor of investigating the possible \emph{relationships of distributed computing systems to physical and biological systems}. The first main motivation for this is the fact that a wide range of physical and biological systems are governed by underlying laws that are essentially \emph{algorithmic}. The second is that the higher-level physical or behavioral properties of such systems are usually the outcome of the coexistence, which may include both cooperation and competition, and constant interaction of \emph{very large numbers of relatively simple distributed entities} respecting such laws. This effort, to the extent that its perspective allows, is expected to promote our understanding on the algorithmic aspects of our (distributed) natural world and to develop innovative artificial systems inspired by them.

Ulam's and von Neuman's Cellular Automata (cf. e.g. \cite{Sc11}), essentially a distributed grid network of automata, have been used as models for self-replication, for modeling several physical systems (e.g. neural activity, bacterial growth, pattern formation in nature), and for understanding emergence, complexity, and self-organization issues. Population Protocols of Angluin \emph{et al.} \cite{AADFP06} were originally motivated by highly dynamic networks of simple sensor nodes that cannot control their mobility. Recently, Doty \cite{Do14} demonstrated their formal equivalence to \emph{chemical reaction networks} (CRNs), which model chemistry in a \emph{well-mixed solution}. Moreover, the \emph{Network Constructors} extension of population protocols \cite{MS14}, showed that a population of finite-automata that interact randomly like molecules in a well-mixed solution and that can establish bonds with each other according to the rules of a common small protocol, can construct arbitrarily complex stable networks \cite{MS14}. In the young area of DNA self-assembly it has been already demonstrated that it is possible to fold long, single-stranded DNA molecules into arbitrary nanoscale two-dimensional shapes and patterns \cite{Ro06}. Recently, an interesting theoretical model was proposed, the \emph{Nubot} model, for studying the complexity of self-assembled structures with active molecular components \cite{WCG13}. This model is inspired by biology's fantastic ability to assemble biomolecules that form systems with complicated structure and dynamics, from molecular motors that walk on rigid tracks and proteins that dynamically alter the structure of the cell during mitosis, to embryonic development where large-scale complicated organisms efficiently grow from a single cell.  Also recently a system was reported that demonstrates programmable self-assembly of complex two-dimensional shapes with a thousand-robot swarm, called the \emph{Kilobot} \cite{RCN14}. This was enabled by creating small, cheap, and simple autonomous robots designed to operate in large groups and to cooperate through local interactions and by developing a collective algorithm for shape formation that is highly robust to the variability and error characteristic of large-scale decentralized systems.

The established and ongoing research seems to have opened the road towards a vision that will further reshape society to an unprecedented degree. This vision concerns our ability to \emph{manipulate matter} via information-theoretic and computing mechanisms and principles. It will be the jump from amorphous information to the \emph{incorporation of information to the physical world}. Information will not only be part of the physical environment: it will constantly interact with the surrounding environment and will have the ability to reshape it. \emph{Matter will become programmable} \cite{GCJM05} which is a plausible future outcome of progress in high-volume nanoscale assembly that makes it feasible to inexpensively produce millimeter-scale units that integrate computing, sensing, actuation, and locomotion mechanisms. This will enable the astonishing possibility of transferring the discrete dynamics from the computer memory black-box to the real world and to achieve a \emph{physical realization of any computer-generated object}. It will have profound implications for how we think about chemistry and materials. Materials will become user-programmed and smart, adapting to changing conditions in order to maintain, optimize or even create a whole new functionality using means that are intrinsic to the material itself. It will even change the way we think about engineering and manufacturing. We will for the first time be capable of building smart machines that adapt to their surroundings, such as an airplane wing that adjusts its surface properties in reaction to environmental variables \cite{Za07}, or even further realize machines that can self-built autonomously.

\subsection{Our Approach}

We imagine here a ``solution'' of automata (also called \emph{nodes} or \emph{processes} throughout the paper), a setting similar to that of Population Protocols and Network Constructors. Due to its highly restricted computational nature and its very local perspective, each individual automaton can practically achieve nothing on its own. However, when many of them cooperate, each contributing its meager computational capabilities, impressive global outcomes become feasible. This is, for example, the case in the Kilobot system, where each individual robot is a remarkably simple artifact that can perform only primitive locomotion via a simple vibration mechanism. Still, when a thousand of them work together, their global dynamics and outcome resemble the complex functions of living organisms. From our perspective, cooperation  involves the capability of the nodes to communicate by interacting in pairs and to bind to each other in an algorithmically controlled way. In particular, during an interaction, the nodes can update their local states according to a small common program that is stored in their memories and may also choose to connect to each other in order to start forming some required structure. Later on, if needed, they may choose to drop their connection, e.g. for rearrangement purposes. We may think of such nodes as the smallest possible programmable pieces of matter. For example, they could be tiny nanorobots or programmable molecules (e.g. DNA strands). Naturally, such elementary entities are not (yet) expected to be equipped with some internal mobility mechanism. Still, it is reasonable to expect that they could be part of some dynamic environment, like a boiling liquid or the human circulatory system, providing an external (to the nodes) interaction mechanism. This, together with the fact that the dynamics of such models have been recently shown to be equivalent to those of CRNs, motivate the idea of regarding such systems as \emph{a solution of programmable entities}. We model such an environment by imagining an \emph{adversary scheduler} operating in discrete steps and selecting in every step a pair of nodes to interact with each other.

Our main focus in this work, building upon the findings of \cite{MS14}, is to further investigate the cooperative structure formation capabilities of such systems. Our first main goal is to introduce a more realistic and more applicable version of network constructors by adjusting some of the abstract parameters of the model of \cite{MS14}. In particular, we introduce some physical (or geometrical) constraints on the connections that the processes are allowed to form. In the network constructors model of \cite{MS14}, there were no such imposed restrictions, in the sense that, at any given step, any two processes were candidates for an interaction, independently of their relative positioning in the existing structure/network. For example, even two nodes hidden in the middle of distinct dense components could interact and, additionally, there was no constraint on the number of active connections that a node could form (could be up to the order of the system). This was very convenient for studying the capability of such systems to self-organize into abstract networks and it helped show that arbitrarily complex networks are in principle constructible. On the other hand, this is not expected to be the actual mechanism of at least the first potential implementations. First implementations will most probably be characterized by physical and geometrical constraints. To capture this in our model, we assume that each device can connect to other devices only via a very limited (finite and independent of the size of the system) number of ports, usually four or six, which implies that, at any given time, a device has only a bounded number of neighbors. Moreover, we further restrict the connections to be always made at unit distance and to be perpendicular to connections of neighboring ports. Though such a model can no longer form abstract networks, it may still be capable of forming very practical 2-dimensional or 3-dimensional shapes. This is also in agreement with natural systems, where the complexity and physical properties of a system are rarely the result of an unrestricted interconnection between entities. 

It can be immediately observed that the universal constructors of \cite{MS14} do not apply in this case. In particular, those constructors cannot be adopted in order to characterize the constructive power of the model considered here. The reason is that they work by arranging the nodes in a long line and then exploiting the fact that connections are elastic and allow any pair of nodes of the line to interact independently of the distance between them. In contrast, no elasticity is allowed in the more local model considered here, where a long line can still be formed but only adjacent nodes of the line are allowed to interact with each other. As a result, we have to develop new techniques for determining the computational and constructive capabilities of our model. The other main novelty of our approach, concerns our attempt to overcome the inability of such systems to detect termination due to their limited global knowledge and their limited computational resources. For example, it can be easily shown that deterministic termination of population protocols can fail even in determining whether there is a single $a$ in an input assignment, mainly because the nodes do not know and cannot store in their memories neither the size of the network nor some upper bound on the time it takes to meet (or to influence or to be influenced by) every other node. To overcome the storage issue, we exploit the ability of nodes to self-assemble into larger structures that can then be used as distributed memories of any desired length. Moreover, we exploit the common (and natural in several cases) assumption that the system is \emph{well-mixed}, meaning that, at any given time, all permissible pairs of node-ports have an equal probability to interact, in order to give \emph{terminating protocols that are correct with high probability}. This is crucial not only because it allows to improve eventual stabilization to eventual termination but, most importantly, because it allows to develop terminating subroutines that can be sequentially composed to form larger modular protocols. Such protocols are more efficient, more natural, and more amenable to clear proofs of correctness, compared to existing protocols that are based on composing all subroutines in parallel and ``sequentializing'' them eventually by perpetual reinitializations. To the best of our knowledge, \cite{MCS12c} is the only work that has considered this issue but with totally different and more deterministic assumptions. Several other papers \cite{AADFP06,AAE08,MS14} have already exploited a uniform random interaction model, but in all cases for analyzing the expected time to convergence of stabilizing protocols and not for maximizing the correctness probability of terminating protocols, as we do here. 

In Section \ref{sec:rw}, we discuss further related literature. Section \ref{sec:model} formally defines the model under consideration and brings together all definitions and basic facts that are used throughout the paper. In Section \ref{sec:basic-con}, we provide direct (stabilizing) constructors for some basic shape construction problems. Section \ref{sec:counting} introduces our technique for counting the size $n$ of the system with high probability. The result of that section (i.e. Theorem \ref{the:count-half}) is of particular importance as it underlies all sequential composition arguments that follow in the paper. In particular, the protocol of Section \ref{sec:counting} is used as a subroutine in our \emph{universal constructors}, presented in Section \ref{sec:gen-con}, establishing that \emph{it is possible to construct with high probability arbitrarily complex shapes (and patterns) by terminating protocols}. These \emph{universality} results are discussed in Section \ref{sec:gen-con}. In Section \ref{sec:replication} we consider the problem of shape self-replication. Finally, in Section \ref{sec:conclusions} we conclude and give further research directions that are opened by our work.

\section{Further Related Work}
\label{sec:rw}

\noindent\textbf{Population Protocols.} Our model for shape construction is strongly inspired by the Population Protocol model \cite{AADFP06} and the Mediated Population Protocol model \cite{MCS11-2}. In the former, connections do not have states. States on the connections were first introduced in the latter. The main difference to our model is that \emph{in those models the focus was on the computation of functions of some input values and not on network construction}. Another important difference is that we allow the edges to choose between \emph{only two possible states} which was not the case in \cite{MCS11-2}. Interestingly, when operating under a uniform random scheduler, population protocols are formally equivalent to \emph{chemical reaction networks} (CRNs) which model chemistry in a \emph{well-mixed solution} \cite{Do14}. CRNs are widely used to describe information processing occurring in natural cellular regulatory networks, and with upcoming advances in synthetic biology, CRNs are a promising programming language for the design of artificial molecular control circuitry. However, CRNs and population protocols can only capture the dynamics of molecular counts and not of structure formation. Our model then may also be viewed as an extension of population protocols and CRNs aiming to capture the stable structures that may occur in a well-mixed solution. From this perspective, our goal is to determine what stable structures can result in such systems (natural or artificial), how fast, and under what conditions (e.g. by what underlying codes/reaction-rules). Most computability issues in the area of population protocols have now been resolved. Finite-state processes on a complete interaction network, i.e. one in which every pair of processes may interact, (and several variations) compute the \emph{semilinear predicates} \cite{AAER07}. Semilinearity persists up to $o(\log\log n)$ local space but not more than this \cite{MNPS11}.  If additionally the connections between processes can hold a state from a finite domain (note that this is a stronger requirement than the on/off that the present work assumes) then the computational power dramatically increases to the commutative subclass of $\rem{NSPACE}(n^2)$ \cite{MCS11-2}. Other important works include \cite{GR09} which equipped the nodes of population protocols with unique ids and \cite{BBCK10} which introduced a (weak) notion of speed of the nodes that allowed the design of fast converging protocols with only weak requirements. For a introductory texts see \cite{AR07,MCS11}.\\ 

\noindent\textbf{Algorithmic Self-Assembly.} There are already several models trying to capture the self-assembly capability of natural processes with the purpose of engineering systems and developing algorithms inspired by such processes. For example, \cite{Do12} proposes to learn how to program molecules to manipulate themselves, grow into machines and at the same time control their own growth. The research area of ``algorithmic self-assembly'' belongs to the field of ``molecular computing''. The latter was initiated by Adleman \cite{Ad94}, who designed interacting DNA molecules to solve an instance of the Hamiltonian path problem. The model that has guided the study in algorithmic self-assembly is the Abstract Tile Assembly Model (aTAM) \cite{Wi98,RW00} and variations. Recently, the Nubot model was proposed \cite{WCG13}, which was another important influence for our work. That model aims at motivating engineering of molecular structures that have complicated active dynamics of the kind seen in living biomolecular systems. It tries to capture the interplay between molecular structure and dynamics. Simple molecular components form assemblies that can grow (exponentially fast, by successive doublings) and shrink, and individual components undergo state changes and move relative to each other. The main result of \cite{WCG13} was that any computable shape of size $\leq n\times n$ can be built in time polylogarithmic in $n$, plus roughly the time needed to simulate a TM that computes whether or not a given pixel is in the final shape.\\

\noindent\textbf{Distributed Network Construction.} To the best of our knowledge, classical distributed computing has not considered the problem of constructing an actual communication network from scratch. From the seminal work of Angluin \cite{An80} that initiated the theoretical study of distributed computing systems up to now, the focus has been more on assuming a given communication topology and constructing a virtual network over it, e.g. a spanning tree for the purpose of fast dissemination of information. Moreover, these models assume most of the time unique identities, unbounded memories, and message-passing communication. Additionally, a process always communicates with its neighboring processes (see \cite{Ly96} for all the details). An exception is the area of geometric pattern formation by mobile robots (cf. \cite{SY99,DFSY10} and references therein). A great difference, though, to our model is that in mobile robotics the computational entities have complete control over their mobility and thus over their future interactions. That is, the goal of a protocol is to result in a desired interaction pattern while in our model the goal of a protocol is to construct a structure while operating under a totally unpredictable interaction pattern. Very recently, the \emph{Amoebot} model, a model inspired by the behavior of amoeba that allows algorithmic research on self-organizing particle systems and programmable matter, was proposed \cite{DGRS13,DDGRS14}. The goal is for the particles to self-organize in order to adapt to a desired shape without any central control, which is quite similar to our objective, however the two models have little in common. In the same work, the authors observe that, in contrast to the considerable work that has been performed w.r.t. to systems (e.g. self-reconfigurable robotic systems), only very little theoretical work has been done in this area. This further supports the importance of introducing a simple, yet sufficiently generic, models for distributed shape construction, as we do in this work.\\

\noindent\textbf{Network Formation in Nature.} Nature has an intrinsic ability to form complex structures and networks via a process known as \emph{self-assembly}. By self-assembly, small components (like e.g. molecules) automatically assemble into large, and usually complex structures (like e.g. a crystal). There is an abundance of such examples in the physical world. Lipid molecules form a cell's membrane, ribosomal proteins and RNA coalesce into functional ribosomes, and bacteriophage virus proteins self-assemble a capsid that allows the virus to invade bacteria \cite{Do12}. Mixtures of RNA fragments that self-assemble into self-replicating ribozymes spontaneously form cooperative catalytic cycles and networks. Such cooperative networks grow faster than selfish autocatalytic cycles indicating an intrinsic ability of RNA populations to evolve greater complexity through cooperation \cite{VMC12}. Through billions of years of prebiotic molecular selection and evolution, nature has produced a basic set of molecules. By combining these simple elements, natural processes are capable of fashioning an enormously diverse range of fabrication units, which can further self-organize into refined structures, materials and molecular machines that not only have high precision, flexibility and error-correction capacity, but are also self-sustaining and evolving. In fact, nature shows a strong preference for bottom-up design.

Systems and solutions inspired by nature have often turned out to be extremely practical and efficient. For example, the bottom-up approach of nature inspires the fabrication of biomaterials by attempting to mimic these phenomena with the aim of creating new and varied structures with novel utilities well beyond the gifts of nature \cite{Zh03}. Moreover, there is already a remarkable amount of work envisioning our future ability to engineer computing and robotic systems by manipulating molecules with nanoscale precision. Ambitious long-term applications include molecular computers \cite{BPSPF10} and miniature (nano)robots for surgical instrumentation, diagnosis and drug delivery in medical applications (e.g. it has very recently been reported that DNA nanorobots could even kill cancer cells \cite{DBC12}) and monitoring in extreme conditions (e.g. in toxic environments). However, the road towards this vision passes first through our ability to discover \emph{the laws governing the capability of distributed systems to construct networks}. The gain of developing such a theory will be twofold: It will give some insight to the role (and the mechanisms) of network formation in the complexity of natural processes and it will allow us engineer artificial systems that achieve this complexity.

\section{A Model of Network Constructors}
\label{sec:model}

There are $n$ nodes. Every node is a finite-state machine taking states from a finite set $Q$. Additionally, every node has a bounded number of ports which it uses to interact with other nodes. In the 2-dimensional (2D) case, there are four ports $p_y$, $p_x$, $p_{-y}$, and $p_{-x}$. For notational convenience and to improve readability we almost exclusively use $u,r,d,l$ (for \emph{up}, \emph{right}, \emph{down}, and \emph{left}, respectively) in place of $p_y$, $p_x$, $p_{-y}$, $p_{-x}$, respectively. Similarly, in the 3-dimensional (3D) case there are 6 ports $p_y$, $p_z$, $p_x$, $p_{-y}$, $p_{-z}$, and $p_{-x}$ (see Figure \ref{fig:node-ports}). Neighboring ports are perpendicular to each other, forming local axes. For example, in the 2-dimensional case, $p_y\perp p_x$, $p_x\perp p_{-y}$, $p_{-y}\perp p_{-x}$, and $p_{-x}\perp p_{y}$. Similar assumptions hold for the 3-dimensional case.  We can imagine the nodes moving passively in a well-mixed solution. An important remark is that the above coordinates are only for local purposes and do not necessarily represent the actual orientation of a node in the system. A node may be arbitrarily rotated so that, for example, its $x$ local coordinate is aligned with the $y$ real coordinate of the system or it is not aligned with any real coordinate. We assume that nodes may interact in pairs, whenever a port of one node is at unit distance and in straight line (w.r.t. to the local axes) from a port of another node. For example, it could be the case that, at some point during execution, the axis of the $p_y$ port of a node $u$ becomes aligned with the axis of the $p_{-x}$ port of another node $v$ and the distance between them is one unit. Then $u$ and $v$ interact and, apart from updating their local states, they can also activate the connection between their corresponding ports. Later on, they can again deactivate the connection if they want.

\begin{figure}[!hbtp]
\centering{
\includegraphics[width=0.35\textwidth]{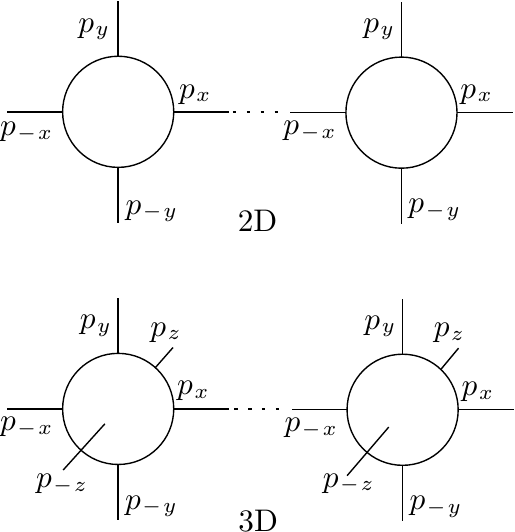}
}
\caption{The top figure depicts the 2D version of the model. Each node has four ports and consecutive ports are perpendicular to each other. Two nodes are interacting, the left one via its $p_x$ port and the right one via its $p_{-x}$ port. The interaction can occur because the distance between the nodes is unit and the corresponding ports are totally aligned (in a straight line). The bottom figure depicts the 3D version of the model. The only difference is an extra $z$ dimension.} \label{fig:node-ports}
\end{figure}

\begin{definition}
A 2D (or 3D) protocol is defined by a 4-tuple $(Q,q_0,Q_{out},\delta)$, where $Q$ is a finite set of \emph{node-states}, $q_0\in Q$ is the \emph{initial node-state}, $Q_{out}\subseteq Q$ is the set of \emph{output node-states}, and $\delta : (Q\times P)\times (Q\times P)\times \{0,1\} \rightarrow Q\times Q\times \{0,1\}$ is the \emph{transition function}, where $P=\{u,r,d,l\}$ (or $P=\{p_y, p_z, p_x, p_{-y}, p_{-z}, p_{-x}\}$, resp.) is the set of \emph{ports}. When required, also a special \emph{initial leader-state} $L_0\in Q$ may be defined.
\end{definition}

If $\delta((a,p_1),(b,p_2),c) = (a^{\prime},b^{\prime},c^{\prime})$, we call $(a,p_1),(b,p_2),c \rightarrow (a^{\prime},b^{\prime},c^{\prime})$ a \emph{transition} (or \emph{rule}). A transition $(a,p_1),(b,p_2),c \rightarrow (a^{\prime},b^{\prime},c^{\prime})$ is called \emph{effective} if $x\neq x^\prime$ for at least one $x\in\{a,b,c\}$ and \emph{ineffective} otherwise. When we present the transition function of a protocol we only present the effective transitions. Additionally, we agree that the \emph{size} of a protocol is the number of its states, i.e. $|Q|$.

The system consists of a population $V$ of $n$ distributed \emph{processes}, called \emph{nodes} when clear from context. Execution of the protocol proceeds in discrete steps. In every step, an unordered (cf. \cite{MS14} for more details and formalism about unordered interactions) pair of node-ports $(v_1,p_1)(v_2,p_2)$ is selected by an \emph{adversary scheduler} and these nodes interact via the corresponding ports and update their states and the state of the edge joining them according to the transition function $\delta$. Before every step $t\geq 1$, there is a shape configuration $G(t)=(V,E(t))$ where $E(t)\subseteq \{(v_1,p_1)(v_2,p_2):v_1,v_2\in V \mbox{ and } p_1,p_2\in P\}$. $(v_1,p_1)(v_2,p_2)\in E(t)$ means that before the $t$th selection of the scheduler the $p_1$ port of node $v_1$ is connected by an active edge to the $p_2$ port of node $v_2$. Observe that not all possible $E(t)$ are valid given the restrictions that connections are made at unit distance and are perpendicular whenever they correspond to consecutive ports of a node. For example, if $(v_1,r)(v_2,l)\in E(t)$ then $(v_1,l)(v_2,r)\notin E(t)$. In general, $E(t)$ is \emph{valid} (or \emph{feasible}) if any component defined by it is a subnetwork of the \emph{2-dimensional grid network with unit distances} (depicted e.g. in Figure \ref{fig:shape1} on page \pageref{fig:shape1}). From now on, we call a 2D (3D) \emph{shape} any connected subnetwork of the 2D (3D) grid network with unit distances. 

The shapes of $E(t)$ also determine the possible selections of the scheduler at step $t$. In particular, $(v_1,p_1)(v_2,p_2)$ can be selected for interaction at step $t$ iff after aligning port $p_1$ of $v_1$ and $p_2$ of $v_2$ at unit distance from each other and vertically to the neighboring ports, the component that would result by activating the connection (that component is the union of the shape of $v_1$ and the shape of $v_2$) is a shape. In particular, there should be no two nodes in the union (one from the first shape and one from the other) that occupy the same position of the grid network (i.e. that one falls over the other). For a simple example, imagine a \huge $\llcorner$ \normalsize shape consisting of three nodes and a vertical line $\mid$ consisting of two nodes. The bottom node of the line is allowed to occupy the missing corner of the \huge $\llcorner$ \normalsize while, on the other hand, the upper node of the line is not allowed (unless rotated) because this would result in the lower node of the line falling over the right node of the \huge $\llcorner$\normalsize . Though in principle a connected component could operate autonomously internally without having to wait for the scheduler to pick a pair of connected nodes to interact, throughout this work, for simplicity and uniformity of the arguments, we also restrict interactions between connected pairs to be controlled only by the scheduler. Notice that any port pair that is connected by an active edge before step $t$ may be selected by the scheduler at step $t$. 

A \emph{configuration} $C$ is a pair $(C_V, C_E)$, where $C_V : V \rightarrow Q$ specifies the state of each node and $C_E: \{(v_1,p_1)(v_2,p_2):v_1,v_2\in V \mbox{ and } p_1,p_2\in P\}\rightarrow \{0,1\}$ specifies the state (inactive or active) of every possible port pair. Observe that if $C_E$ is the edge configuration before step $t$, then $E(t)=C_E^{-1}[1]=\{r_1r_2:C_E(r_1r_2)=1\}$. We write $C\rightarrow C^{\prime}$ if \emph{$C^\prime$ is reachable in one step from $C$} (meaning via a single interaction that is permitted on $C$). We say that $C^{\prime}$ is \emph{reachable} from $C$ and write $C\rsa C^{\prime}$, if there is a sequence of configurations $C=C_{0},C_{1},\ldots,C_{t}=C^{\prime}$, such that $C_{i}\rightarrow C_{i+1}$ for all $i$, $0\leq i <t$. An \emph{execution} is a finite or infinite sequence of configurations $C_{0},C_{1},$ $C_{2},\ldots$, where $C_{0}$ is an initial configuration and $C_{i}\rightarrow C_{i+1}$, for all $i\geq 0$. A \emph{fairness condition} is imposed on the adversary to ensure the protocol makes progress. An infinite execution is \emph{fair} if for every pair of configurations $C$ and $C^{\prime}$ such that $C\rightarrow C^{\prime}$, if $C$ occurs infinitely often in the execution then so does $C^{\prime}$. In most cases, we assume that interactions are chosen by a \emph{uniform random scheduler} which, in every step $t$, selects independently and uniformly at random one of the interactions permitted by $E(t)$. Note that the uniform random scheduler is fair with probability 1. Wherever no such probabilistic scheduling assumption is made, every execution of a protocol will by definition considered to be fair. Random schedulers are particularly useful when one wants to analyze the running time of protocols. In this work, we mainly exploit them in order to devise terminating protocols that are correct with high probability (abbreviated w.h.p.), always meaning in this work with probability at least $1-1/n^c$ for some constant $c\geq 1$.

We define the \emph{output of a configuration} $C$ as the set of shapes $G(C)=(V_s,E_s)$ where $V_s=\{u\in V: C_V(u)\in Q_{out}\}$ and $E_s=C_E^{-1}[1]\cap\{(v_1,p_1)(v_2,p_2):v_1,v_2\in V_s \mbox{ and } p_1,p_2\in P\}$. In words, the output shapes of a configuration consist of those nodes that are in output states and those edges between them that are active. Throughout this work, we are interested in obtaining a single shape as the final output of the protocol (see, for an example, the black nodes and the connections between them in Figure \ref{fig:shape4} on page \pageref{fig:shape4}). As already mentioned, our main focus will be on terminating protocols. In this case, we assume a set $Q_{halt}\subseteq Q$ in place of $Q_{out}$. The only difference is that for all $q_{halt}\in Q_{halt}$, every rule containing $q_{halt}$ is ineffective. In contrast, states in $Q_{out}$ may have effective interactions which we guarantee (by design) to cease eventually resulting in the stabilization of the final shape. 

\begin{definition}
We say that an execution of a protocol on $n$ processes \emph{constructs (stably constructs) a shape} $G$, if it terminates (stabilizes, resp.) with output $G$.
\end{definition}

Every 2D shape $G$ has a unique minimum 2D rectangle $R_G$ enclosing it. $R_G$ is a shape with its nodes labeled from $\{0,1\}$. The nodes $G$ are labeled $1$, the nodes in $V(R_G)\bs V(G)$ are labeled $0$, and all edges are active. It is like filling $G$ with additional nodes and edges to make it a rectangle. In fact, as we shall see in Section \ref{sec:replication}, $R_G$ can also be constructed by a protocol, given $G$. The dimensions of $R_G$ are defined by $h_G$, which is the horizontal distance between a leftmost node and a rightmost node of the shape (x-dimension), and $v_G$, which is the vertical distance between a highest and a lowest node of the shape (y-dimension). Let also $max\_dim_G\coleq\max\{h_G,v_G\}$ and $min\_dim_G\coleq\min\{h_G,v_G\}$. Then $R_G$ can be extended by $max\_dim_G-min\_dim_G$ extra rows or columns, depending on which of its dimensions is smaller, to a $max\_dim_G\times max\_dim_G$ square $S_G$ enclosing $G$ (we mean here a $\{0,1\}$-labeled square, as above, in which $G$ can be identified). Observe, that such a square is not unique. For example, if $G$ is a horizontal line of length $d$ (i.e. $h_G=d$ and $v_G=1$) then it is already equal to $R_G$ and has to be extended by $d-1$ rows to become $S_G$. These rows can be placed in $d$ distinct ways relative to $G$, but all these squares have the same size $max\_dim_G\times max\_dim_G$ denoted by $|S_G|$.

A ($j$-dimensional) \emph{shape language} $\cl$ is a subset of the set of all possible ($j$-dimensional) shapes. We restrict our attention here to shape languages that contain a unique shape for each possible maximum dimension of the shape. In this case, it is equivalent, and more convenient, to translate $\cl$ to a language of labeled squares. In particular, we define in this work a \emph{shape language} $\cl$ by providing for every $d\geq 1$ a single $d\times d$ square with its nodes labeled from $\{0,1\}$. Such a square may also be defined by a $d^2$-sequence $S_d=(s_0,s_1,\ldots,s_{d^2-1})$ of bits or \emph{pixels}, where $s_j\in\{0,1\}$ corresponds to the $j$-th node as follows: We assume that the pixels are indexed in a ``zig-zag'' fashion, beginning from the bottom left corner, moving to the right until the bottom right corner is encountered, then one step up, then to the left until the node above the bottom left corner is encountered, then one step up again, then right, and so on (see the directed path in Figure \ref{fig:shape2} on page \pageref{fig:shape2}). The shape $G_d$ defined by $S_d$, called \emph{the shape of $S_d$}, is the one induced by the nodes of the square that are labeled 1 and throughout this work we assume that $max\_dim_{G_d}=d$.  

For simulation purposes, we also need to introduce appropriate shape-constructing Turing Machines (TMs). We now describe such a TM $M$: $M$'s goal is to construct a shape on the pixels of a $\sqrt{n}\times\sqrt{n}$ square, which are indexed in the zig-zag way described above. $M$ takes as input an integer $i\in\{0,1,\ldots,n-1\}$ and the size $n$ or the dimension $\sqrt{n}$ of the square (all in binary) and decides whether pixel $i$ should belong or not to the final shape, i.e. if it should be \emph{on} or \emph{off}, respectively. \footnote{If the TM is not provided with the square size together with the pixel, then it can only compute uniform/symmetric shapes that are independent of $n$. Such a shape could for example be one that has every even pixel \emph{on} and every odd pixel \emph{off}. But such shapes rarely satisfy the connectivity condition. For example, it is not clear how to activate all the leftmost pixels of the square by a uniform TM, because such a TM should somehow guess that pixel $2\sqrt{n}-1$ should be accepted without knowing $n$ and given that all pixels in $[1,2\sqrt{n}-2]$ must be rejected. So, it seems more natural to consider TMs that apart from the pixel index are also provided with $n$ or $\sqrt{n}$ (if the latter is more convenient) in binary. Now, it is straightforward how to resolve the acceptance of only the leftmost pixels of the square. The TM every time accepts the input-pixel $i$ iff $i=2k\sqrt{n}-1$, for some $k\geq 1$, or $i=2k\sqrt{n}$, for some $k\geq 0$. Observe that $2k\sqrt{n}$ can always be computed because the TM is also provided with $\sqrt{n}$ in its input.} Moreover, in accordance to our definition of a shape, the construction of the TM, consisting of the pixels that $M$ accepts (as \emph{on}) and the active connections between them, should be \emph{connected} (i.e. it should be a single shape).

\begin{definition} \label{def:computable}
We say that a shape language $\cl=(S_1,S_2,S_3,\ldots)$ is \emph{TM-computable} or \emph{TM-constructible} in space $f(d)$, if there exists a TM $M$ (as defined above) such that, for every $d\geq 1$, when $M$ is executed on the pixels of a $d\times d$ square results in $S_d$ (in particular, on input $(i,d)$, where $0\leq i\leq d^2-1$, $M$ gives output $S_d[i]$), by using space $O(f(d))$ in every execution.
\end{definition}

\begin{definition}
We say that a protocol $\ca$ \emph{constructs a shape language $\cl$ with useful space $g(n)\leq n$}, if $g(n)$ is the greatest function for which: (i) for all $n$, every execution of $\ca$ on $n$ processes constructs a shape $G\in \cl$ \footnote{$G$ is the shape of a labeled square $S\in\cl$ in case $\cl$ is defined in terms of such squares.} of order at least $g(n)$ (provided that such a $G$ exists) and, additionally, (ii) for all $G\in \cl$ there is an execution of $\ca$ on $n$ processes, for some $n$ satisfying $|V(G)|\geq g(n)$, that constructs $G$. Equivalently, we say that \emph{$\ca$ constructs $\cl$ with waste $n-g(n)$}.
\end{definition}

\section{Basic Constructions}
\label{sec:basic-con}

We give in this section protocols for two very basic shape construction problems, the spanning line problem and the spanning square problem. Both constructions are very useful because they organize the nodes in a way that is convenient for TM simulations that exploit the whole distributed memory as a tape. Keep in mind that the protocols of this section are \emph{stabilizing} (that is eventually the output shape stops changing) and not terminating. Our technique that allows for terminating constructions will be introduced in Section \ref{sec:counting}.

\subsection{Global Line}

We begin by presenting a protocol for the spanning line problem. Assume that there is initially a unique leader in state $L_r$ (we typically use `$L$' for the states of a leader to distinguish from the left port `$l$') and all other nodes are in state $q_0$. A protocol that constructs a spanning line is described by the effective rules $(L_i,i),(q_0,j),0\ra (q_1,L_{\bar{j}},1)$ for all $i,j\in\{u,r,d,l\}$ where $\bar{j}$ denotes the port opposite to port $j$. In words, initially the leader waits to meet a $q_0$ via its right port. Assume that it meets port $j$ of a $q_0$. Then the connection between them becomes activated and the leader takes the place of the $q_0$ leaving behind a $q_1$. Moreover, the new leader is now in state $L_{\bar{j}}$ indicating that it is now waiting to expand the line towards the port that is opposite to the one that is already active, which guarantees that a straight line will be formed. We could even have a simplified version of the form $(L,r),(q_0,l),0\ra (q_1,L,1)$. This is a little slower, because now an effective interaction, and a resulting expansion of the line, only occurs when the $r$ port of the leader meets the $l$ port of a $q_0$.

\subsection{$\sqrt{n}\times\sqrt{n}$ Square}
\label{subsec:stab-square}

We now give two protocols for the spanning square problem. We assume for simplicity that the square root of $n$ is integer. We again begin from the case where there is a preelected unique leader in state $L_u$ and all other nodes are initially in state $q_0$.

\floatname{algorithm}{Protocol}
\renewcommand{\algorithmiccomment}[1]{// #1}
\begin{algorithm}[!h]
  \caption{\emph{Square}}\label{prot:square}
  \begin{algorithmic}
    \medskip
    \State $Q=\{L_u,L_r,L_d,L_l,q_0,q_1\}$
    \State $\delta$: 
    \begin{align*}
    (L_u,u),(q_0,d),0&\ra (q_1,L_r,1)\\
    (L_r,r),(q_0,l),0&\ra (q_1,L_d,1)\\
    (L_d,d),(q_0,u),0&\ra (q_1,L_l,1)\\
    (L_l,l),(q_0,r),0&\ra (q_1,L_u,1)\\
    (L_u,u),(q_1,d),0&\ra (L_l,q_1,1)\\
    (L_r,r),(q_1,l),0&\ra (L_u,q_1,1)\\
    (L_d,d),(q_1,u),0&\ra (L_r,q_1,1)\\
    (L_l,l),(q_1,r),0&\ra (L_d,q_1,1)
    \phantom{\hspace{10cm}}
    \end{align*}
    \State \Comment {All transitions that do not appear have no effect}
  \end{algorithmic}
\end{algorithm} 

We now describe the idea that Protocol \ref{prot:square} implements. The protocol first constructs a $2\times 2$ square. When it is done, the leader is at the bottom-right corner and is in state $L_d$. This can only cause the attachment of a free $q_0$ from below. When this occurs, the leader moves on the new node and tries to move to the left. This will occur by the attachment of another free node from the left this time. When this occurs, the leader moves on the new node and tries to move up. But this time the up movement cannot succeed because the leader is below the bottom-left corner of the square. Instead the leader activates the connection with that corner and tries to move another step left. When it succeeds, it tries again to move up, which now can occur because it has now moved outside the left boundary of the $2\times 2$ square. In general, whenever the leader is at the left (the up, right, and down cases are symmetric) of the already constructed square it tries to move to the right in order to walk above the square. If it does not succeed, it is because it has not yet moved over the upper boundary, so it activates the edge to the right, takes another step up and then tries again to move to the right. In this way, the leader always grows the square perimetrically and in the clockwise direction.

We next use turning marks to simplify and speed up the turning process. The unique leader begins in state $L_d^2$. Now, instead of always trying to turn, the leader turns only when it meets special marks left by the previous phase near the corners of the square. When it meets such a mark, the leader introduces the new corner and a new mark adjacent to that corner to be found during the next phase, and then makes a turn (see Figure \ref{fig:square2}). A difference to the previous protocol is that now several of the nodes of the new perimeter may remain disconnected for a while from their internal neighbors (i.e. those belonging to the internal perimeter constructed in the previous phase). However, rules of the form $(q_1,i),(q_1,\bar{i}),0\ra (q_1,q_1,1)$ guarantee that these nodes eventually become connected. A disadvantage of this approach is that the structure may be less ``rigid'' than the previous one as long as several $(q_1,q_1)$ connections are not yet established. The protocol is formally presented in Protocol \ref{prot:square2}.

\begin{figure}[!hbtp]
\centering{
\includegraphics[width=0.65\textwidth]{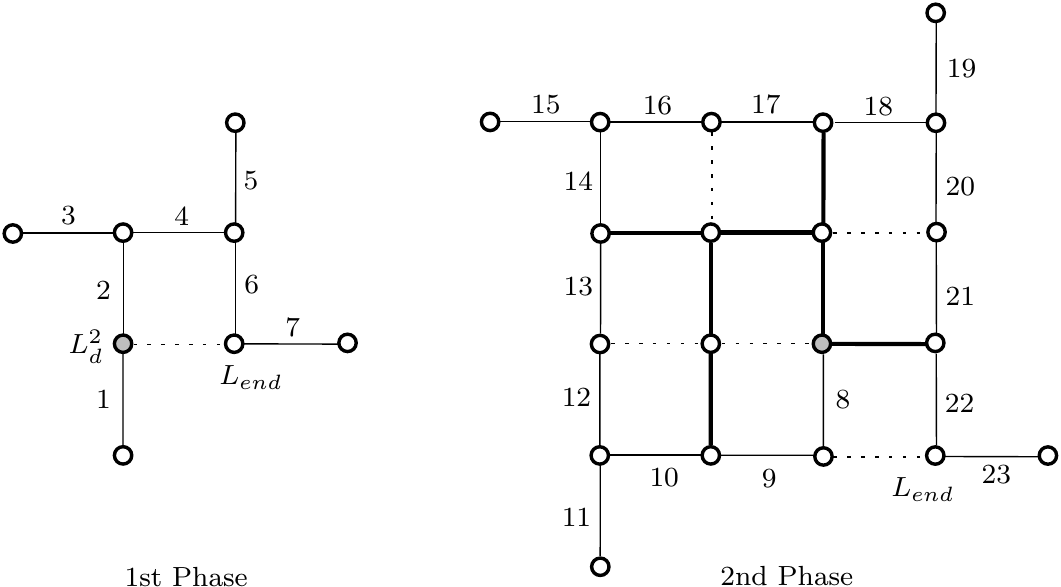}
}
\caption{The first two phases of Protocol \ref{prot:square2}. Gray nodes indicate the starting point of each phase. Edge labels indicate the order by which the square grew during the phase. The nodes labeled $L_{end}$ are the points at which each of the phases ends. The unlabeled solid edges of Phase 2 indicate the shape that preexisted from Phase 1. The nodes attached at ``times'' $1,3,5,7$ of Phase 1 and $11,15,19,23$ of Phase 2 are the turning marks that will be exploited for easier turning by the leader in the subsequent phase. Dotted edges are edges that have not be activated yet but will for sure be activated eventually resulting in a more ``rigid'' structure.} \label{fig:square2}
\end{figure}

\floatname{algorithm}{Protocol}
\renewcommand{\algorithmiccomment}[1]{// #1}
\begin{algorithm}[!h]
  \caption{\emph{Square2}}\label{prot:square2}
  \begin{algorithmic}
    \medskip
    \State $Q=\{L_i,L_i^2,L_i^3,L_i^4,L_{end},q_0,q_1\}$, for all $i\in\{u,r,d,l\}$
    \State $\delta$: 
    \begin{align*}
    (L_d^2,d),(q_0,u),0&\ra (L_u^1,q_1,1)\\
    (L_l^2,l),(q_0,r),0&\ra (L_r^1,q_1,1)\\
    (L_u^2,u),(q_0,d),0&\ra (L_d^1,q_1,1)\\
    (L_r^2,r),(q_0,l),0&\ra (L_{end},q_1,1)\\\\
    (L_u^1,u),(q_0,d),0&\ra (q_1,L_l^2,1)\\
    (L_r^1,r),(q_0,l),0&\ra (q_1,L_u^2,1)\\
    (L_d^1,d),(q_0,u),0&\ra (q_1,L_r^2,1)\\
    (L_r^1,u),(q_0,d),0&\ra (q_1,L_l^2,1)\\\\      
    (L_{end},d),(q_0,u),0&\ra (q_1,L_l,1)\\\\
	(L_l,l),(q_0,r),0&\ra (q_1,L_l,1)\\
	(L_l,l),(q_1,r),0&\ra (q_1,L_l^3,1)\\
	(L_u,u),(q_0,d),0&\ra (q_1,L_u,1)\\
	(L_u,u),(q_1,d),0&\ra (q_1,L_u^3,1)\\
	(L_r,r),(q_0,l),0&\ra (q_1,L_r,1)\\
	(L_r,r),(q_1,l),0&\ra (q_1,L_r^3,1)\\
	(L_d,d),(q_0,u),0&\ra (q_1,L_d,1)\\
	(L_d,d),(q_1,u),0&\ra (q_1,L_d^3,1)\\\\	
	(L_l^3,l),(q_0,r),0&\ra (q_1,L_d^4,1)\\
    (L_u^3,u),(q_0,d),0&\ra (q_1,L_l^4,1)\\
    (L_r^3,r),(q_0,l),0&\ra (q_1,L_u^4,1)\\
    (L_d^3,d),(q_0,u),0&\ra (q_1,L_r^4,1)\\\\      
	(L_d^4,d),(q_0,u),0&\ra (L_u,q_1,1)\\
    (L_l^4,l),(q_0,r),0&\ra (L_r,q_1,1)\\
    (L_u^4,u),(q_0,d),0&\ra (L_d,q_1,1)\\
    (L_r^4,r),(q_0,l),0&\ra (L_{end},q_1,1)\\\\
	(q_1,i),(q_1,\bar{i}),0&\ra (q_1,q_1,1) \text{ for all } i\in\{u,r,d,l\}, \text{ where } \bar{i} \text{ denotes the opposite port of } i\\
	(L_u,r),(q_1,l),0&\ra (L_u,q_1,1)\\
	(L_r,d),(q_1,u),0&\ra (L_r,q_1,1)\\
	(L_d,l),(q_1,r),0&\ra (L_d,q_1,1)\\
	(L_l,u),(q_1,d),0&\ra (L_l,q_1,1)			    		    
    \phantom{\hspace{10cm}}
    \end{align*}
  \end{algorithmic}
\end{algorithm} 

The unique leader assumption is in all the above cases not necessary.

\section{Probabilistic Counting}
\label{sec:counting}

In this section, we consider the problem of counting $n$. In particular, we assume a uniform random scheduler and we want to give protocols that always terminate but still w.h.p. count $n$ correctly. The importance of such protocols is further supported by the fact that we cannot guarantee anything much better than this. In particular, observe that if we require a population protocol to always terminate and additionally to always be correct, then we immediately obtain an impossibility result. It is easy to see this by imagining a system in which a unique leader interacts with the other nodes (there are no interactions between non-leaders and no connections are ever activated). Any fair execution $s_1$ of a protocol in a population of size $n$ in which the leader outputs $n$ and terminates can appear as an ``unfair'' prefix of a fair execution $s^\prime=s_1s_2$ on a population of size $n^\prime>n$. This is a contradiction because in $s^\prime$ the leader must again terminate and output $n$ even though $n^\prime\neq n$. The main reason is that $|s_1|$ is finite and independent of $n$; it only depends on the maximum ``depth'' of a chain of rules of the protocol leading to termination. This implies that in $s^\prime$ the leader terminates before interacting with all other nodes.

In Section \ref{subsec:leader-counting}, we present a population protocol with a unique leader that solves w.h.p. the counting problem and always terminates. To the best of our knowledge, this is the first protocol of this sort in the relevant literature. All probabilistic protocols that have appeared so far, like e.g. those in \cite{AADFP06,AAE08}, are not terminating but stabilizing and the high probability arguments concern their time to convergence. Additionally, this protocol is crucial because all of our generic constructors, that are developed in Section \ref{sec:gen-con}, are terminating by assuming knowledge of $n$ (stored distributedly on a line of length $\log n$). They obtain access to this knowledge w.h.p. by executing the counting protocol as a subroutine. Finally, knowing $n$ w.h.p. allows to develop protocols that exploit sequential composition of (terminating) subroutines, which makes them much more natural and easy to describe than the protocols in which all subroutines are executed in parallel and perpetual reinitializations is the only means of guaranteeing eventual correctness (the latter is the case e.g. in \cite{GR09,MCS11-2,MS14}, but not in \cite{MCS12c} which was the first extension to allow for sequential composition based on some non-probabilistic assumptions). Then in Section \ref{subsec:impossibility} we comment on the possibility of dropping the unique leader assumption. In particular, we conjecture that in general it is impossible to solve the problem if all nodes are identical and we present some evidence supporting this. Finally, in Section \ref{subsec:uids} we establish that if the nodes have unique ids then it is possible to solve the problem without a unique leader.  

\subsection{Fast Probabilistic Counting With a Leader}
\label{subsec:leader-counting}

Keep in mind that in order to simplify the discussion, a sort of population protocol is presented here. So, there are no ports, no geometry, and no activations/deactivations of connections. In every step, a uniform random scheduler selects equiprobably one of the $n(n-1)/2$ possible node pairs, and the selected nodes interact and update their states according to the transition function. The only difference from the classical population protocols is that a distinguished leader node has unbounded local memory (of the order of $n$). In Section \ref{subsec:counting-line}, we will adjust the protocol to make it work in our model.\\

\noindent\textbf{Counting-Upper-Bound Protocol:} There is initially a unique leader $l$ and all other nodes are in state $q_0$. Assume that $l$ has two $n$-counters in its memory, initially both set to 0. So, the state of $l$ is denoted as $l(r_0,r_1)$, where $r_0$ is the value of the first counter (call the corresponding counter $R_0$) and $r_1$ the value of the second counter (call it $R_1$), $0\leq r_0,r_1\leq n$. The rules of the protocol are 
\begin{align*}
(l(r_0,r_1),q_0)&\ra (l(r_0+1,r_1),q_1),\\
(l(r_0,r_1),q_1)&\ra (l(r_0,r_1+1),q_2), \mbox{ and}\\
(l(r_0,r_1),\cdot)&\ra (halt,\cdot) \mbox{ if } r_0=r_1. 
\end{align*}
It is worth reminding that, for the time being, we have disregarded edge-states and, therefore, the rules of the protocol only specify how the states of the nodes are updated. Observe that $r_0$ counts the number of $q_0$s in the population while $r_1$ counts the number of $q_1$s. Initially, there are $n-1$ $q_0$s and no $q_1$s. Whenever $l$ interacts with a $q_0$, $r_0$ increases by 1 and the $q_0$ is converted to $q_1$. Whenever $l$ interacts with a $q_1$, $r_1$ increases by 1 and the $q_1$ is converted to $q_2$. The process terminates when $r_0=r_1$ for the first time. We also give to $r_0$ an initial head start of $b$, where $b$ can be any desired constant. So, initially we have $r_0=b$, $r_1=0$ and $i=\#q_0=n-b-1$, $j=\#q_1=b$ (this can be easily implemented in the protocol by having the leader convert $b$ $q_0$s to $q_1$s as a preprocessing step). So, we have two competing processes, one counting $q_0$s and the other counting $q_1$s, the first one begins with an initial head start of $b$ and the game ends when the second catches up the first. We now prove that when this occurs the leader will almost surely have already counted at least half of the nodes.

\begin{theorem} \label{the:count-half}
The above protocol halts in every execution. Moreover, when this occurs, w.h.p. it holds that $r_0 \geq n/2$.
\end{theorem}
\begin{proof}
Recall that the scheduler is a uniform random one, which, in every step, selects independently and uniformly at random one of the $n(n-1)/2$ possible interactions. Recall also that the random variable $i$ denotes the number of $q_0$s and $j$ denotes the number of $q_1$s in the configuration, where initially $i=n-b-1$ and $j=b$. Observe also that all the following hold: $j=r_0-r_1$, $r_0\geq r_1$, because every conversion of a $q_1$ to $q_2$ must have been first counted by $R_0$ as a conversion of a $q_0$ to $q_1$, $r_1=(n-1)-(i+j)$, and $r_0+r_1$ is equal to the number of effective interactions (see Figure \ref{fig:counting-random-variables}).

\begin{figure}[!hbtp]
\centering{
\includegraphics[width=0.73\textwidth]{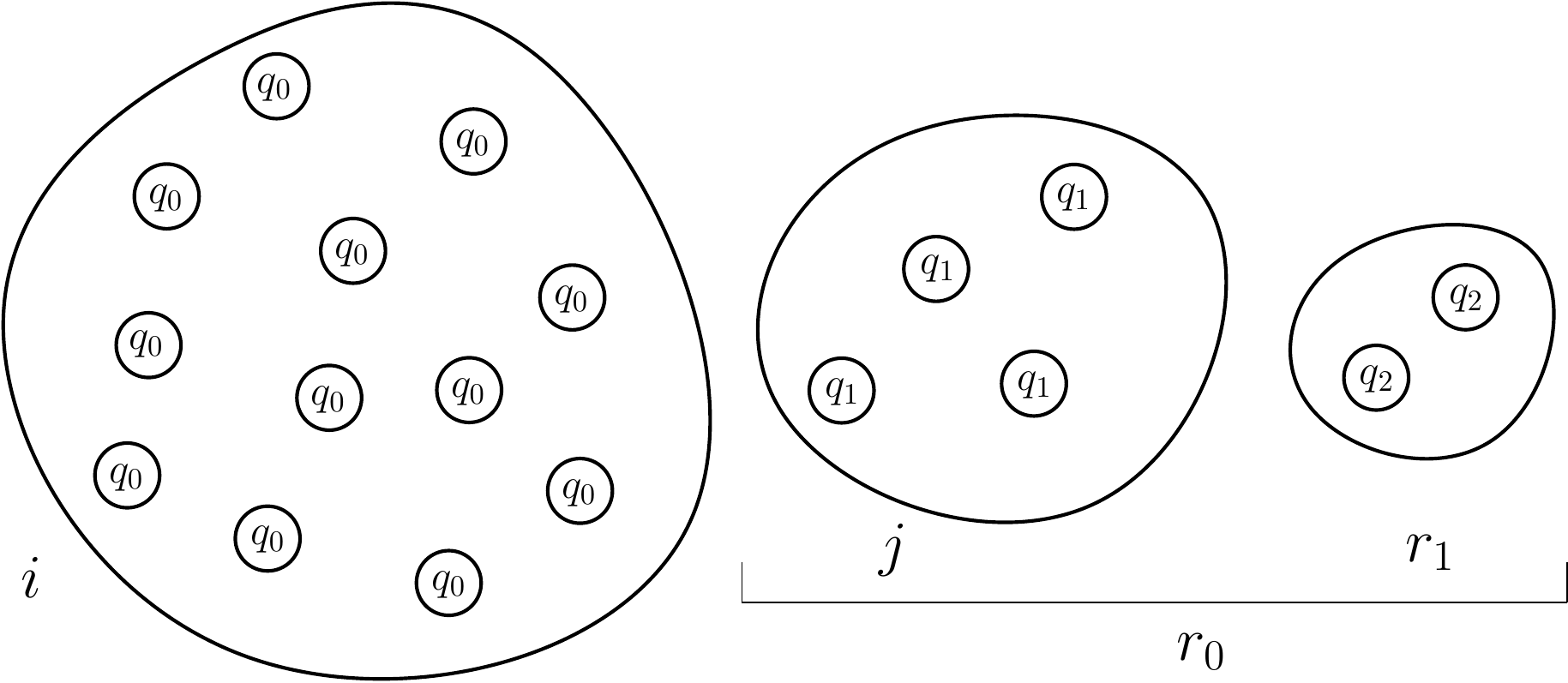}
}
\caption{A configuration of the system (excluding the leader). The number of $q_0$s remaining is denoted by $i$. The number of $q_1$s introduced so far is denoted by $j$. The value of the counter $R_1$ is equal to the number of $q_1$s encountered so far by the leader, which is in turn equal to the number of $q_2$s introduced, and is denoted by $r_1$. The value of the counter $R_0$ is equal to the number of $q_0$s encountered, which is equal to the number of $q_1$s and $q_2$s introduced and is denoted by $r_0$.} \label{fig:counting-random-variables}
\end{figure}

We focus only on the effective interactions (we also disregard the halting interaction), which are always interactions between $l$ and $q_0$ or $q_1$. Given that we have an effective interaction, the probability that it is an $(l,q_0)$ is $p_{ij}=i/(i+j)$ and the probability that it is an $(l,q_1)$ is $q_{ij}=1-p_{ij}=j/(i+j)$. This random process may be viewed as a random walk (r.w.) on a line with $n+1$ positions $0,1,\ldots,n$ where a particle begins from position $b$ and there is an absorbing barrier at $0$ and a reflecting barrier at $n$. The position of the particle corresponds to the difference $r_0-r_1$ of the two counters which is equal to $j$. Observe now that if $j\geq n/2$ then $r_0-r_1\geq n/2\Rightarrow r_0\geq n/2$, so it suffices to consider a second absorbing barrier at $n/2$. The particle moves forward (i.e. to the right) with probability $p_{ij}$ and backward with probability $q_{ij}$ (see Figure \ref{fig:random-walk-ij}). This is a ``difficult'' random walk because the transition probabilities not only depend on the position $j$ but also on the sum $i+j$ which decreases in time. In particular, the sum decreases whenever an $(l,q_1)$ interaction occurs, in which case a $q_1$ becomes $q_2$. That is, whenever the random walk returns to some position $j$ of the line, its transition probabilities have changed (because every leaving and returning involves at least on step to the left, which decreases the sum). Observe also that, in our case, the duration of the random walk can be at most $n-b$, in the sense that if the particle has not been absorbed after $n-b$ steps then we have success. The reason for this is that $n-b$ effective interactions imply that $r_0+r_1= n$, but as $r_0\geq r_1$, we have $r_0\geq n/2$. In fact, $r_0\geq n/2\Leftrightarrow$ $j+r_1\geq n/2$. We are interested in upper bounding $\P[\text{failure}]=\P[\text{reach 0 before } r_0 \geq n/2 \text{ holds}]$, which is in turn upper bounded by the probability of reaching 0 before reaching $n/2$ and before $n-b$ effective interactions have occurred (this is true because, in the latter event, we have disregarded some winning conditions like, for example, guaranteed winning in $(n/2)+r_1$ effective interactions, in which case we have winning in only $(n/2)+r_1$ effective interactions and $j$ having become at most $(n/2)-r_1$). It suffices to bound the probability of reaching 0 before $n$ effective interactions have occurred.

\begin{figure}[!hbtp]
\centering{
\includegraphics[width=0.73\textwidth]{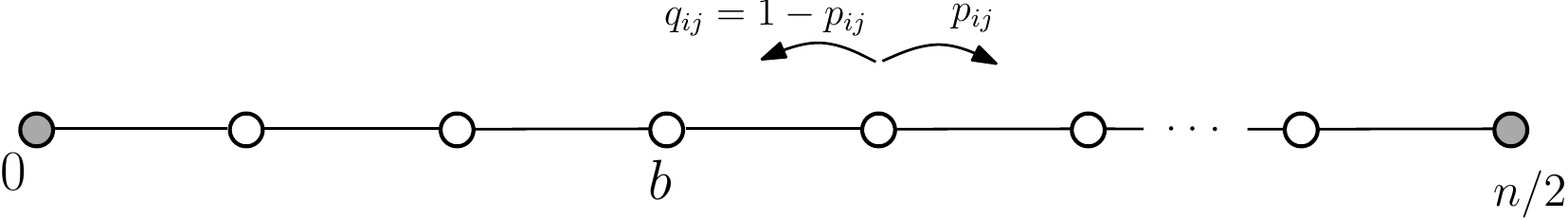}
}
\caption{A random walk modeling of the probabilistic process that the Counting-Upper-Bound protocol implements. A particle begins from position $b$. The position $j$ of the particle corresponds to the difference between $r_0$ and $r_1$. Forward movement corresponds to an increment of $r_0$ and backward movement corresponds to an increment of $r_1$. Absorption at 0 corresponds to $r_1$ becoming equal to $r_0$ and thus to termination (and to failure if this occurs before $r_0\geq n/2$ holds). Absorption at $n/2$ corresponds to $r_0$ becoming at least $n/2$ (before being absorbed at 0) and thus to success.} \label{fig:random-walk-ij}
\end{figure}

Thus, we have $r_0+r_1\leq n$ but $r_1 \leq r_0\Rightarrow 2r_1\leq r_0+r_1$, thus $2r_1\leq n\Rightarrow$ $r_1\leq n/2\Rightarrow$ $(n-1)-(i+j)\leq n/2\Rightarrow$ $i+j\geq (n/2)-1=n^\prime$. And if we set $n^\prime=(n/2)-1$ we have $i+j\geq n^\prime$. Moreover, observe that when $r_0+r_1=n+1$ we have $n+1=r_0+r_1\leq 2r_0\Rightarrow$ $r_0\geq n/2$. In summary, during the first $n$ effective interactions, it holds that $i+j\geq n^\prime=(n/2)-1$ and when interaction $n+1$ occurs it holds that $r_0\geq n/2$, that is, if the process is still alive after time $n$, then $r_0$ has managed to count up to $n/2$ and the protocol has succeeded.

Now, $i+j\geq n^\prime$ implies that $p_j\geq (n^\prime-j)/n'$  and $q_j\leq j/n^\prime$ so that now the probabilities only depend on the position $j$. This new walk is the well-studied Ehrenfest random walk coming from the theory of brownian motion. Imagine gas molecules that move about randomly in a container which is divided into two halves symmetrically by a partition. A hole is made in the partition to allow the exchange of molecules between the subcontainers. Suppose there are n molecules in the container. Think of the partitions as two urns, I and II, containing balls labeled $1$ through $n$. Molecular motion can be modeled by choosing a number between 1 and $n$ at random and moving the corresponding ball from the urn it is presently in to the other. This is a historically important physical model, known as the Ehrenfest model of diffusion, introduced in \cite{EE07} in the early days of statistical mechanics to study thermodynamic equilibrium. So, the probability of failure of our counting protocol is asymptotically equivalent to the probability that urn I becomes empty in the first $n$ steps assuming that it initially contains $b$ balls. This walk has been studied by Kac in \cite{Ka47} who, among other things, proved that the mean recurrence time is $((R+k)!(R-k)!/(2R)!)2^{2R}$ (\cite{Ka47}, page 386). If we set $k=-R$ so that the initial position is $R+k=0$, then this evaluates to $2^{2R} = 2^{n/2}$, because $2R$ is the total length of the line. This shows that, even if we begin from position 0 instead of b, the recurrence time is expected to be huge and we do not expect the walk to return to 0 and fail in only $n$ effective steps. In the sequel, we turn this into the desired high probability argument.

We will reduce the Ehrenfest walk to one in which the probabilities do not depend on $j$. We first further restrict our walk, this time to the prefix $[0,b]$ of the line. In this part, it holds that $j\leq b$ which implies that $p\geq (n^\prime-b)/n^\prime$  and $q\leq b/n^\prime$. Now we set $p= (n^\prime-b)/n^\prime$  and $q= b/n^\prime$. Observe that this may only increase the probability of failure, so the probability of failure of the new walk is an upper bound on the probability of failure of our original walk. Recall that initially the particle is on position $b$. Imagine now an absorbing barrier at 0 and another one at $b$. Whenever the r.w. is on $b-1$ it will either return to $b$ before reaching 0 or it will reach 0 (and fail) before returning to $b$. So, we now have a r.w. with $b+1$ positions, where positions 0 and $b$ are absorbing and due to symmetry it is equivalent to assume that the particle begins from position 1, moves forward with probability $p^\prime=q$, backward with probability $q^\prime=p$, and it fails at $b$. Thus, it is equivalent to bound $\P[$reach $b$ before 0 (when beginning from position 1)$]$. This is the probability of winning in the classical ruin problem analyzed e.g. in \cite{Fe68} page 345. If we set $x=q^\prime/p^\prime=p/q=(n^\prime-b)/b$ we have that:
\begin{align*}
\P[\text{reach } b \text{ before } 0]&=1-\frac{x^b-x}{x^b-1}=\frac{x-1}{x^b-1}\\
&\leq \frac{x}{x^b-1}\approx \frac{1}{x^{b-1}}\\
&\approx \frac{1}{n^{b-1}}.
\end{align*}

Thus, whenever the original walk is on $b-1$, the probability of reaching 0 before reaching $b$ again, is at most $1/n^{b-1}$. Now assume that we repeat the above walk $n$ times, i.e. we place the particle on $b-1$, play the game, then if it returns to $b$ we put again the particle on $b-1$ and play the game again, and so on. From Boole-Bonferroni inequality, we have that:

\begin{align*}
\P[\text{fail at least once}]&\leq \sum_{m=1}^{n} \P[\text{fail at repetition }m]\\
&\leq \sum_{m=1}^{n} \frac{1}{n^{b-1}}=\frac{n}{n^{b-1}}\\
&=\frac{1}{n^{b-2}}.
\end{align*}

In summary, even if the protocol was restricted to disregard counter differences that are greater than $b$, still with probability at least $1-1/n^c$ (for constant $c=b-2$) the protocol has not terminated after at least $n$ effective interactions, which in turn implies that the leader has counted at least half of the nodes.
\qed
\end{proof}

\begin{remark}
For the Counting-Upper-Bound protocol to terminate, it suffices for the leader to meet every other node twice. This takes twice the expected time of a \emph{meet everybody} (cf. \cite{MS14}), thus the expected running time of Counting-Upper-Bound is $O(n^2\log n)$ (interactions). 
\end{remark}

\begin{remark}
When the Counting-Upper-Bound protocol terminates, w.h.p. the leader knows an $r_0$ which is between $n/2$ and $n$. So any subsequent routine can use directly this estimation and pay in an \emph{a priori} waste which is at most half of the population. In practice, this estimation is expected to be much closer to $n$ than to $n/2$ (in all of our experiments for up to 1000 nodes, the estimation was always close to $(9/10)n$ and usually higher). On the other hand, if we want to determine the exact value of $n$ and have no \emph{a priori} waste then we can have the leader wait an additional large polynomial (in $r_0$) number of steps, to ensure that the leader has met every other node w.h.p. (observe e.g. that the last unvisited node requires an expected number of $\Theta(n^2)$ steps to be visited).
\end{remark}

\subsection{Impossibility of Counting Without a Leader}
\label{subsec:impossibility}

An immediate question is whether the unique leader assumption of Theorem \ref{the:count-half} can be dropped. Unfortunately, the answer to this question seems to be negative. In particular, it seems that any protocol in which all nodes begin from the same state may have some node terminate with (at least) constant probability having participated in only a constant number of interactions. This implies that with constant probability the protocol terminates without having estimated any non-constant function of $n$.

Nodes again have a set of states $Q$ and we also assume that they have unbounded private local memories. These memories are for internal purposes only and their contents are not communicated to other nodes. For example, a node $u$ could maintain $|Q|$ counters, each counting the number of times the corresponding state has been encountered so far by $u$. We focus on protocols that always terminate (i.e for every $n\geq n_0$, for some finite $n_0$) and we want them to compute something w.h.p., e.g. the node that first terminates to know an upper bound on $n$ w.h.p..

\begin{conjecture} \label{conj:count-impossibility}
Let $\ca$ be a protocol as above. Then, as $n$ grows, there is (at least) a constant probability that some node terminates having interacted only a constant number of times.  
\end{conjecture}

We now give some evidence why the above conjecture, that excludes the possibility of protocols that count $n$ w.h.p., seems to be true. First of all, observe that a protocol apart from the usual transition function $\delta:Q\times Q\ra Q\times Q$ that updates the communicating states has also a function $\gamma:Q\times S\ra S$ that updates the internal memory based on the encountered states. We focus on deterministic $\gamma$ and in this case the internal state from $S$ after $k$ interactions only depends on the observed sequence $Q^k$ of encountered states (because the initial state $q_0$ is always the same for all nodes). Every protocol $\ca$ that always terminates, essentially defines a property $L_{\ca}\subseteq Q^*$ consisting of those observed sequences that make a node terminate (the remaining sequences do not cause termination). Moreover, as the protocol does not know $n$, an $s_0\in L_{\ca}$ of minimum length has length that is independent of $n$ (it could only be a function of $|Q|$). Observe that for every population size $n$, if $s_0$ is observed by some node $u$ as a prefix of its interaction pattern (i.e. in its first $|s_0|$ interactions) then $u$ terminates while having participated in only $|s_0|$ interactions, which is a constant number independent of $n$. What it seems to hold is that, for every $n\gg n_0$ and every such fixed $s_0$, there is (at least) a constant probability that some node observes $s_0$. In particular we believe that it might be possible to prove the following set of arguments provided that $n\gg n_0$:
\begin{enumerate}
\item With constant probability a configuration is reached, in which every state $q\in Q$ has multiplicity $\Theta(n)$ (that is appears on $\Theta(n)$ distinct nodes).
\item With constant probability the multiplicities of all states remain $\Theta(n)$ for $\Theta(n)$ steps.
\item While (2) holds, with constant probability one of the $\Theta(n)$ nodes, let it be $u$, whose state is $q_0$, interacts $|s_0|$ times.
\end{enumerate}
If the above hold, then it follows that $u$ may observe $s_0$ with constant probability, in which case $u$ will terminate having interacted only a constant (i.e. $|s_0|$) number of times. The reason for this is that in its $i$th interaction, for all $1\leq i\leq |s_0|$, $u$ observes the $i$th state of $s_0$, let it be $q_i$, with probability $(\#q_i \mbox{ in the population})/\Theta(n)$. As, by (2), the numerator is also $\Theta(n)$, for all $q_i\in Q$, the resulting probability is constant. Unfortunately, we have not yet been able to turn this into a formal proof. 

\subsection{Counting Without a Leader but With UIDs}
\label{subsec:uids}

Now nodes have unique ids from a universe $\cu$. Nodes initially do not know the ids of other nodes nor $n$. The goal is again to count $n$ w.h.p.. All nodes execute the same program and no node can initially act as unique leader, because nodes do not know which ids from $\cu$ are actually present in the system. Nodes have unbounded memory but we try to minimize it, e.g. if possible store only up to a constant number of other nodes ids. We show that under these assumptions, the counting problem can be solved without the necessity of a unique leader.

\subsubsection{A Simple Protocol}

We first show feasibility by a very simple protocol which guarantees that the nodes w.h.p. count $n$ by paying a large termination time.

\textbf{Protocol:} Every node $u$ remembers its first $b$ interactions, where $b$ is a predetermined constant. In particular, it maintains a vector v$_u$ of length $b$ and in every interaction $i$, $1\leq i\leq b$, with a node with $v$ it sets v$_u(i)\leftarrow id_v$. Also $u$ counts the number of distinct nodes that it has interacted with so far, by placing their ids in an $A_u$ list. Initially $A_u\leftarrow \{$v$_u\}\cup\{id_u\}$ and in every interaction with a node $v$, $u$ sets $A_u\leftarrow A_u\cup\{id_v\}$. Moreover, after interaction $b$, $u$ keeps track of the ids encountered in every $b$ consecutive interactions, in another vector v$^\prime_u$ of length $b$, initially empty. Whenever a sequence of length $b$ is recorded (i.e. the vector is full), if v$_u$=v$^\prime_u$ then $u$ outputs $|A_u|$ and terminates, otherwise $u$ clears the contents of v$^\prime_u$ and starts recording the next $b$ interactions.

\begin{theorem}
When a node $u$ terminates in the above protocol, w.h.p. $|A_u|=n$. The expected termination time is $b(n-1)^b=\Theta(n^b)$. 
\end{theorem}
\begin{proof}
Given the initial sequence of length $b$, i.e. v$_u$, the probability that a sequence of $b$ consecutive interactions observed by $u$ is equal to v$_u$ is $1/(n-1)^b$ and thus the expected time for this to occur is $b(n-1)^b$ interactions of $u$ and as $u$ participates on average every $n$ steps, it is a total of $bn(n-1)^b$ interactions. But there are $n$ nodes doing the same independently of one another, thus the actual expected time for one of the nodes to terminate is $b(n-1)^b=\Theta(n^b)$. On the other hand, the expected time for any of the nodes to meet every other node is only $\Theta(n\log n)$.  
\qed
\end{proof}

\subsubsection{An Improved Protocol}

We now give a protocol that improves the expected time to termination still guaranteeing correct counting w.h.p.. The idea is to have the node with the maximum id in the system to perform the same process as the unique leader in the protocol with no ids of Theorem \ref{the:count-half}. Of course, initially all nodes have to behave as if they were the maximum (as they do not know in advance who the maximum is). By comparing ids during an interaction and deactivating the smaller one, we can easily guarantee that eventually only the maximum, $u_{max}$, will remain active. Moreover, it is clear that $u_{max}$ can always win the other nodes in every interaction, so we easily ensure that its process is not affected by the other nodes. However, we must also guarantee that no other node ever terminates (with sufficiently large probability) early, giving as output a wrong count.

\textbf{Informal description:} Every node $u$ has a unique id $id_u$ and tries to simulate the behavior of the unique leader of the protocol of Theorem \ref{the:count-half}. In particular, whenever it meets another node for the first time it wants to mark it once and the second time it meets that node it wants to mark it twice, recording the number of first-meetings and second-meetings in two local counters. The problem is that now many nodes may want to mark the same node. One idea, of course, could be to have a node remember all the nodes that have marked it so far but we want to avoid this because it requires a lot of memory and communication. Instead, we allow a node to only remember a single other node's id at a time.  Every node tries initially to increase its first-meetings counter to $b$ so that it creates an initial $b$ head start of this counter w.r.t. the other. Every node that succeeds starts executing its main process. The main idea is that whenever a node $u$ interacts with another node that either has or has been marked by an id greater than $id_u$, $u$ becomes deactivated and stops counting. This guarantees that only $u_{max}$ will forever remain active. Moreover, every node $u$ always remembers the maximum id that has marked it so far, so that the probabilistic counting process of a node $u$ can only be affected by nodes with id greater than $id_u$ and as a result no one can affect the counting process of $u_{max}$. Protocol \ref{prot:count-ids} puts all these together formally and Theorem \ref{the:count-ids} shows that this process correctly simulates the counting process of Theorem \ref{the:count-half}, thus providing w.h.p. an upper bound on $n$. 

\floatname{algorithm}{Protocol}
\renewcommand{\algorithmiccomment}[1]{// #1}
\begin{algorithm*}
  \caption{Counting with UIDs}\label{prot:count-ids}
  \begin{algorithmic}[1]
    \Require Every node $u$ has a unique id $id_u$ and maintains a $(belongs,marked)$ pair, a $(count1,count2)$ pair, and a variable $active$, where $belongs\in \cu\cup\{\perp\}$ initially $\perp$, $marked\in\{0,1,2\}$ initially $0$, $count1,count2\in\bbbn_{\geq 0}$ initially $count1=count2=0$ and $active\in\{0,1\}$ initially $1$. All nodes know a predetermined constant $b>0$. The following is the code for every interaction of $u,v$ with $id_u>id_v$.
    \If {$active_v=1$}
        \State $active_v\leftarrow 0$
    \EndIf
    \If {$active_u=1$}
        \If {$belongs_v=\perp$ or $\perp\neq belongs_v<id_u$}
            \State $belongs_v\leftarrow id_u$
	    \State $marked_v\leftarrow 1$
	    \State $count1_u\leftarrow count1_u+1$
        \EndIf
        \If {$\perp\neq belongs_v>id_u$}
	    \State $active_u\leftarrow 0$
        \EndIf
	\If {$belongs_v=id_u$ and $marked_v=1$ and $count1_u\geq b$}
	    \State $marked_v\leftarrow 2$
	    \State $count2_u\leftarrow count2_u+1$
	    \If {$count1_u=count2_u$}
	        \State $u$ halts and outputs $2\cdot count1_u$
	    \EndIf
        \EndIf
    \EndIf
  \end{algorithmic}
\end{algorithm*}

\begin{theorem} \label{the:count-ids}
When a node $u$ in Protocol \ref{prot:count-ids} halts, w.h.p. it holds that $u=u_{max}$ and that $2\cdot count1_u \geq n$.
\end{theorem}
\begin{proof}
We first show that $u_{max}$ simulates the probabilistic process of the unique leader $l$ of Theorem \ref{the:count-half}. Recall that in the protocol of Theorem \ref{the:count-half}, all other nodes are initially $q_0$ and when $l$ meets a $q_0$ it makes it $q_1$ and when it meets a $q_1$ it makes it $q_2$, every time counting in the corresponding counter. First, observe that $u_{max}$ is never deactivated, i.e. $active_{u_{max}}=1$ forever, because it never interacts with a greater id nor with a node that belongs to a greater id than its own. It suffices to show that when $u_{max}$ meets a node for the first time it marks it once (simulating a $q_0$ to $q_1$ conversion), when it meets a node for the second time it marks it twice (simulating a $q_1$ to $q_2$ conversion), and that no other node can ever alter the nodes marked by $u_{max}$. When $u_{max}$ interacts with a node $v$ for the first time, then either $belongs_v=\perp$ or $\perp\neq belongs_v < max\_id$. So, in this case it marks $v$ once by setting $marked_v\leftarrow 1$, $belongs_v\leftarrow max\_id$, and records this by increasing $count1_{u_{max}}$ by one. From now on, no other active node $w\neq u_{max}$ can ever affect the state of $v$, because for every such $w$ it holds that $id_w<belongs_v=max\_id$ and the only effect in this case is the deactivation of $w$. The second time that $u_{max}$ interacts with $v$, it still holds that $belongs_v=id_{u_{max}}$($=max\_id$) and $marked_v=1$, and $u_{max}$ marks $v$ for a second time by setting $marked_v\leftarrow 2$ and records this by incrementing $count2_{u_{max}}$ by one. Again, $v$ still belongs to $max\_id$ and no other node can ever affect its state. We conclude that if we were only interested in $u_{max}$'s output then, by Theorem \ref{the:count-half}, this would w.h.p. be an upper bound on $n$. 

However, observe that not only $u_{max}$ but also the other nodes execute a similar process and it could be the case that one of them terminates early (and before $u_{max}$) giving as output a wrong count. We now show that this is not the case. Take any node $w$ with $id_w<max\_id$. Consider the partition of $V\bs\{w\}$ into the sets $S_{w,0}$, $S_{w,1}$, and $S_{w,2}$ of nodes which $w$ has not marked yet, has marked once, and has marked twice, respectively. $S_{w,1}$ cannot increase without $w$ being involved, so the only possibility that may increase the rate of growth of $count2_w$ w.r.t. $count1_w$ is when a node $v\in S_{w,0}$ gets marked by a node with id greater than $id_w$, because such a $v$ can no longer contribute to $count1_w$. However, observe that every such $v$ will from now on forever satisfy $belongs_v>id_w$, because $belongs_v$ can only increase and every interaction of $w$ with such a $v$ will result in the deactivation of $w$. This implies that the ``success'' events of $w$ (those corresponding to a $count1_w$ increment) have now been partitioned into increment events and deactivation events. So, if $w$ ever fails to increment $count1_w$ due to an interference of some $u$ with $id_u>id_w$ on some $v\in S_{w,0}$, the effect is the deactivation of $w$, which clearly does not allow $w$ to continue with unfavorable probabilities. 
\qed
\end{proof}

\section{Generic Constructors}
\label{sec:gen-con}

In this section, we give a characterization for the class of constructible 2D shape languages. In particular, we establish that shape constructing TMs (defined in Section \ref{sec:model}), can be simulated by our model and therefore we can realize their output-shape in the actual distributed system. To this end, we begin in Section \ref{subsec:counting-line} by adapting the Counting-Upper-Bound protocol of Section \ref{sec:counting} to work in our model. The result is, again w.h.p., a line of length $\Theta(\log n)$ with a unique leader containing $n$ in binary. Then, in Section \ref{subsec:repl-square} the leader exploits its knowledge of $n$ to construct a $\sqrt{n}\times\sqrt{n}$ square. In the sequel (Section \ref{subsec:simulation}), it simulates the TM on the square $n$ distinct times, one for each pixel of the square. Each time, the input provided to the TM is the index of the pixel and $\sqrt{n}$, both in binary. Each simulation decides whether the corresponding pixel should be \emph{on} or \emph{off}. When all simulations have completed, the leader releases in the solution, in a systematic way, the connected shape consisting of the \emph{on} pixels and the active edges between them. The connections of all other (\emph{off}) pixels become deactivated and the corresponding nodes become free (isolated) nodes in the solution. 

\subsection{Storing the Count on a Line}
\label{subsec:counting-line}

We begin by adapting the Counting-Upper-Bound protocol of Theorem \ref{the:count-half} so that when the protocol terminates the final correct count is stored distributedly in binary on an active line of length $\log n$.\\

\noindent\textbf{Counting-on-a-Line Protocol:} The probabilistic process that is being executed is essentially the same as that of the Counting-Upper-Bound protocol. Again the protocol assumes a unique leader that forever controls the process. A difference now is that every node has four ports (in the 2D case). The leader operates as a TM that stores the $r_0$ and $r_1$ counters in its tape in binary. The $i$th cell of the tape has two components, one storing the $i$th bit of $r_0$ and the other storing the $i$th bit of $r_1$. We say that the tape is \emph{full}, if the bits of all $r_0$ components of the tape are set to 1. The tape of the TM is the active line that the leader has formed so far, each node in the line implementing one cell of the tape. Initially the tape consists of a single cell, stored in the memory of the unique leader node. As in Counting-Upper-Bound, the leader first tries to obtain an initial advantage of $b$ for the $r_0$ counter. To achieve the advantage, the leader does not count the $q_1$s that it interacts with until it holds that $r_0\geq b$. Observe that the initial length of the tape is not sufficient for storing the binary representation of $b$ (of course $b$ is constant so, in principle, it could be stored on a single node, however we prefer to keep the description as uniform as possible). In order to resolve this, the leader does the following. Whenever it meets the left port of a $q_0$ from its right port, if its tape is not full yet, it switches the $q_0$ to $q_1$, leaving it free to move in the solution, and increases the $r_0$ counter by one. To increase the counter, it freezes the probabilistic process (that is, during freezing it ignores all interactions with free nodes), and starts moving on its tape, which is a distributed line attached to its left port. After incrementing the counter, the leader keeps track of whether the tape is now full and then it moves back to the right endpoint of the line to unfreeze and continue the probabilistic process. On the other hand, if the tape is full, it binds the encountered $q_0$ to its right by activating the connection between them (thus increasing the length of the tape by one), then it reorganizes the tape, it again increases $r_0$ by one, and finally moves back to the right endpoint to continue the probabilistic process. This time, the leader also records that it has bound a $q_0$ that should have been converted to $q_1$. This \emph{debt} is also stored on the tape in another counter $r_2$. Whenever the leader meets a $q_2$, if $r_2\geq 1$, it converts $q_2$ to $q_1$ and decreases $r_2$ by one. So, $q_2$s may be viewed as a \emph{deposit} that is used to pay back the debt. In this manner, the $q_0$s that are used to form the tape of the leader are not immediately converted to $q_1$ when first counted. Instead, the missing $q_1$s are introduced at a later time, one after every interaction of the leader with a $q_2$, and all of them will be introduced eventually, when a sufficient number of $q_2$s will become available. Finally, whenever the leader interacts with the left port of a $q_1$ from its right port, it freezes, increases the $r_1$ counter by one (observe that $r_0\geq r_1$ always holds, so the length of the tape is always sufficient for $r_1$ increments), and checks whether $r_0=r_1$. If equality holds, the leader terminates, otherwise it moves back to the right endpoint and continues the process. Correctness is captured by the following lemma.

\begin{lemma} \label{lem:counting-line}
Counting-on-a-Line protocol terminates in every execution. Moreover, when the leader terminates, w.h.p. it has formed an active line of length $\log n$ containing $n$ in binary in the $r_0$ components of the nodes of the line (each node storing one bit).
\end{lemma}
\begin{proof}
We begin by showing that the probabilistic process of the Counting-Upper-Bound protocol is not negatively affected in the Counting-on-a-Line protocol. This implies that the high probability argument of Theorem \ref{the:count-half} holds also for Counting-on-a-Line (in fact it is improved). First of all, observe that the four ports of the nodes introduce more choices for the scheduler in every step. However, these new choices, if treated uniformly, result in the same multiplicative factor for both the ``positive'' (an $(l,q_0)$ interaction) and the ``negative'' (an $(l,q_1)$ interaction) events, so the probabilities of the process are not affected at all by this. Moreover, neither the debt affects the process. The reason is that the only essential difference w.r.t. to the process is that the conversion of some counted $q_0$s to the corresponding $q_1$s is delayed. But this only decreases the probability of early termination and thus of failure. It remains to show that not even a single $q_1$ remains forever as debt, because, otherwise, some executions of the protocol would not terminate. The reason is that the protocol cannot terminate before converting all the $q_1$s plus the debt to $q_2$. To this end, observe that the line of the leader has always length $\lfloor\lg r_0\rfloor+1$, thus $r_2\leq \lfloor\lg r_0\rfloor$, because the debt is always at most the length of the line excluding the initial leader. So, at least $r_0-\lfloor\lg r_0\rfloor$ nodes have been successfully converted from $q_0$ to $q_1$ which implies that there is an eventual deposit of at least $r_0-\lfloor\lg r_0\rfloor$ nodes in state $q_2$. These $q_2$s are not immediately available, but they will for sure become available in the future, because every interaction of the leader with a $q_1$ results in a $q_2$. Finally, observe that $r_0-\lfloor\lg r_0\rfloor\geq \lfloor\lg r_0\rfloor$ holds for all $r_0\geq 1$ (to see this, simply rewrite it as $r_0/2\geq \lfloor\lg r_0\rfloor$). Thus, $r_0-\lfloor\lg r_0\rfloor\geq r_2$, which means that the eventual deposit is not smaller than the debt, so the protocol eventually pays back its debt and terminates.
\qed
\end{proof}

\subsection{Constructing a $\sqrt{n}\times\sqrt{n}$ Square}
\label{subsec:repl-square}

We now show how to organize the nodes into a spanning square, i.e. a $\sqrt{n}\times\sqrt{n}$ one. As we did in Section \ref{subsec:stab-square}, we again assume for simplicity that $\sqrt{n}$ is integer. Observe that now the leader has $n$ stored in its line. Our construction exploits this knowledge and this makes it essentially different than the constructions of Section \ref{subsec:stab-square}. Moreover, knowledge of $n$ allows the protocol to terminate after constructing the square and to know that the square has been successfully constructed, a fact that was not the case in the stabilizing constructions of Section \ref{subsec:stab-square}. The following protocol assumes that the guarantee of Lemma \ref{lem:counting-line} is provided somehow and based on this assumption we will show that it works correctly in every execution (this is in contrast to the high probability argument of Lemma \ref{lem:counting-line}). This means that given the guarantee, the protocol that constructs the square is always correct. Of course, if we take the composition of Counting-on-a-Line that provides the guarantee and the protocol that constructs the square based on the guarantee, the resulting protocol is again correct w.h.p., however we still allow the possibility that some other deterministic (even centralized) preprocessing provides the required guarantee.\\

\noindent\textbf{Square-Knowing-\emph{n} Protocol:} The initial leader $L$ first computes $\sqrt{n}$ on its line by any plausible algorithm (observe that the available space for computing the square root is exponential in the binary representation of $n$, which is the input to the algorithm, because, if needed, the leader can expand its line up to length $n$). In principle, it is not necessary to use additional space, because the leader can execute one after the other the multiplications $1\cdot 1$, $2\cdot 2$, $3\cdot 3$, $\ldots$ in binary until the result becomes equal to $n$. Each of these operations can be executed in the initial $\log n$ space of the line of the leader. The time needed, though exponential in the binary representation of $n$, is still linear in the population size $n$. Now that the leader also knows $\sqrt{n}$, it expands its line to the right by attaching free nodes to make its length $\sqrt{n}$. Then it exploits the down ports to create a replica of its line. The replica has also length $\sqrt{n}$ and has its own leader but in a distinguished state $L_s$. This new line plays the role of a \emph{seed} that starts creating other self-replicating lines of length $\sqrt{n}$. In particular, the seed attaches free nodes to its down ports, until all positions below the line are filled by nodes and additionally all horizontal connections between those nodes are activated. Then it introduces a leader $L_r$ to one endpoint of the replica and starts deactivating the vertical connections to release the new line of length $\sqrt{n}$. These lines with $L_r$ leaders are \emph{totally self-replicating}, meaning that their children also begin in state $L_r$. The initial leader $L$ waits until the up ports of a non-seed replica $r$ become totally aligned with the down ports of the square segment that has been constructed so far. So, initially it waits until a replica becomes attached to the lower side of its own line. When this occurs, it activates all intermediate vertical connections to make the construction rigid and increments a row-counter by one (initially 0) and moves to the new lowest row. If at the time of attachment $r$ was in the middle of an incomplete replication, then there will be nodes attached to the down ports of $r$. $L$ releases all these nodes, by deactivating the active connections of $r$ to them, and then waits for another non-seed replica to arrive. When the row-counter becomes equal to $\sqrt{n}-1$, the leader for the first time accepts the attachment of the seed to its construction and when the seed is successfully attached the leader terminates. This completes the construction of the $\sqrt{n}\times\sqrt{n}$ square. See Figures \ref{fig:square-terminating1} and \ref{fig:square-terminating2} for illustrations.

The reason for attaching the seed last, and in particular when no further free nodes have remained, is that otherwise self-replication could possibly cease in some executions. Observe also that we have allowed the $L$-leader to accept the attachment of a replica to the square segment even though the replica may be in the middle of an incomplete replication. This is important in order to avoid reaching a point at which some free lines are in the middle of incomplete replications but there are no further free nodes for any of them to complete. For a simple example, consider the seed and a replica $r$ and $\sqrt{n}$ free nodes (all other nodes have been attached to the square segment). It is possible that $\sqrt{n}-1$ of the free nodes become attached to the seed and the last free node becomes attached to $r$. We have overcome this deadlock by allowing $L$ to accept the attachment of $r$ to the square segment. When this occurs, the free node will be released and eventually it will be attached to the last free position below the seed.

\begin{figure}[!hbtp]
\centering{
\includegraphics[width=0.65\textwidth]{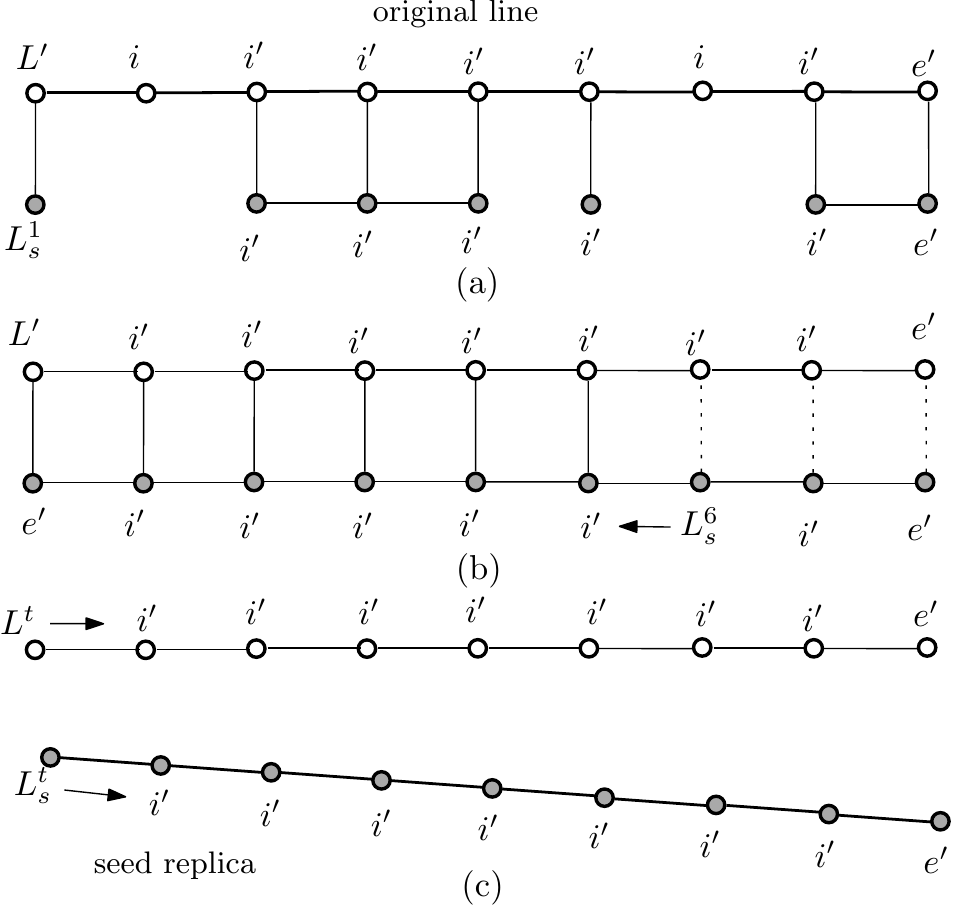}
}
\caption{(a) Several free nodes have already been attached to the original line. Some of them have already activated some horizontal connections forming some segments of the replica. (b) The leader ($L^\prime$) of the original line remains blocked while the leader ($L_s^6$) of the replica has detected that the replica is ready for detachment. It has already detached the three rightmost nodes and keeps moving to the left until it reaches the left endpoint and detaches the whole replica. (c) The seed replica has been released in the solution. The leader ($L^t$) of the original line has waken up and is restoring the nodes of its line to their original states. When it finishes (that is, when it will have traversed the whole line and have returned to the left endpoint), it will go to state $L_{start}$ to start the square formation process. Similarly, the leader ($L_s^t$) of the seed replica is setting the nodes of its line to their normal $i$ and $e$ states, so that they start accepting the attachment of other nodes in order to create non-seed replicas.} \label{fig:square-terminating1}
\end{figure}

\begin{figure}[!hbtp]
\centering{
\includegraphics[width=0.65\textwidth]{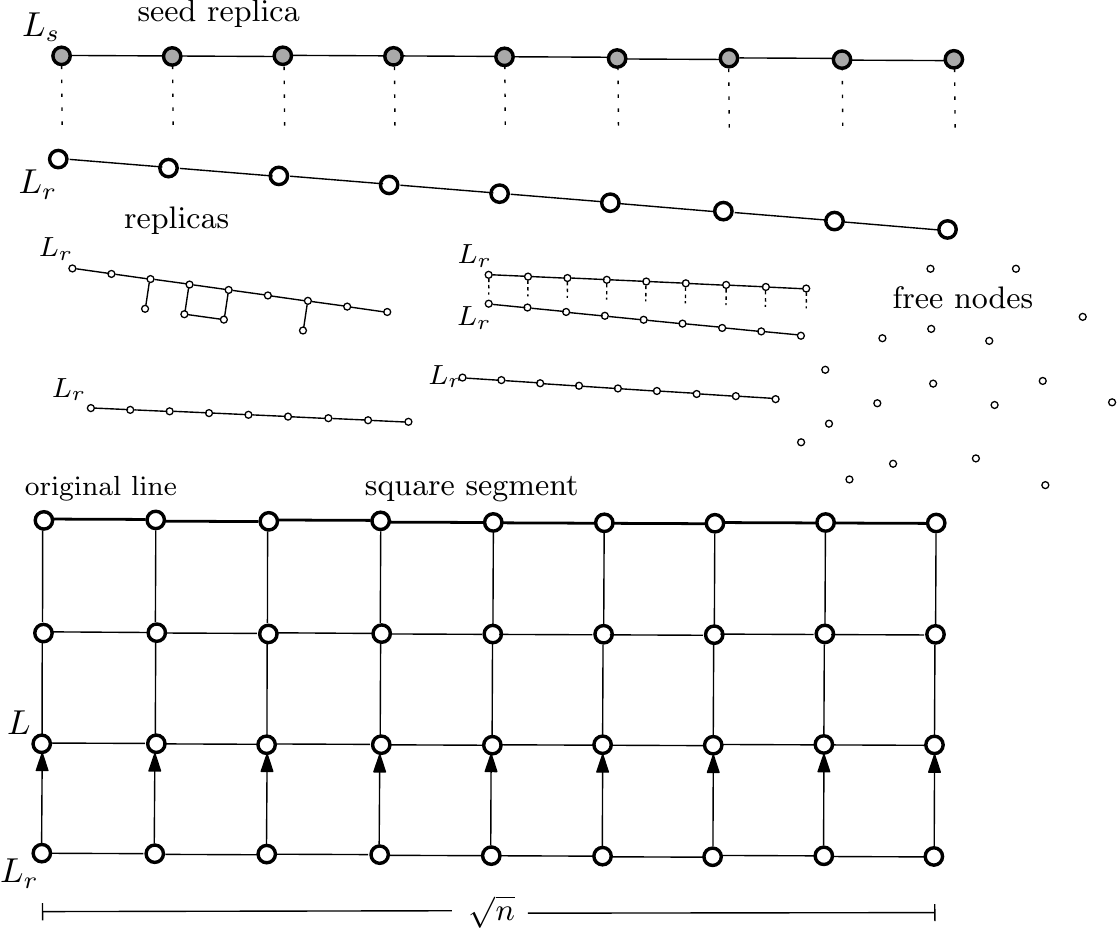}
}
\caption{The seed at the top has created another replica which has just been released in the solution. Below it, some additional replicas appear. One of them is in the middle of a replication that has not completed yet. There are also several nodes that are still free. At the bottom appears the square segment that has been constructed so far. The original line of the $L$-leader is the one at the top of the rectangle. The other rows below it have been formed by replicas that have been attached to the segment in previous steps. The $L$-leader keeps waiting at the bottom left corner for new replicas to arrive. One such has just arrived and will be attached to the segment.} \label{fig:square-terminating2}
\end{figure}

We now give, in Protocol \ref{prot:line-replication}, one of the possible codes for the replication process of the original leader's line that creates the seed. The other replication processes, i.e. of $L_s$ to $L_r$ and of $L_r$ to $L_r$, are almost identical to this one. Without loss of generality we assume that the original leader's line has state $L$ on its left endpoint, $e$ on its right endpoint, and every other internal node of the line is in state $i$. All other (free) nodes are in state $q_0$.

\floatname{algorithm}{Protocol}
\renewcommand{\algorithmiccomment}[1]{// #1}
\begin{algorithm}[!h]
  \caption{\emph{Line-Replication}}\label{prot:line-replication}
  \begin{algorithmic}
    \medskip
    \State $Q=\{L,L^\prime,L_s^j,L_s^t,L_s^{t^{\prime}},L_s^{t^{\dprime}},L^t,L^{t^{\prime}},L^{t^{\dprime}},$ $L_{start},i,i^{\prime},e,e^{\prime}\}$, $j\in\{1,2,\ldots,7\}$
    \State $\delta$: 
    \begin{align*}
    (L,d),(q_0,u),0&\ra (L^\prime,L_s^1,1)\\
    (i,d),(q_0,u),0&\ra (i^\prime,i^\prime,1)\\
    (e,d),(q_0,u),0&\ra (e^\prime,e^\prime,1)\\
    (i^\prime,r),(i^\prime,l),0&\ra (i^\prime,i^\prime,1)\\
    (i^\prime,r),(e^\prime,l),0&\ra (i^\prime,e^\prime,1)\\
    (L_s^1,r),(i^\prime,l),0&\ra (e^\prime,L_s^2,1)\\    			    		    
    (L_s^2,r),(i^\prime,l),\cdot&\ra (i^\prime,L_s^2,1)\\    			    		    
    (L_s^2,r),(e^\prime,l),\cdot&\ra (i^\prime,L_s^3,1)\\
    (L_s^3,u),(e^\prime,d),1&\ra (L_s^4,e^\prime,0)\\
    (i^\prime,r),(L_s^4,l),1&\ra (L_s^5,e^\prime,1)\\
    (L_s^5,u),(i^\prime,d),1&\ra (L_s^6,i^\prime,0)\\
    (i^\prime,r),(L_s^6,l),1&\ra (L_s^5,i^\prime,1)\\
    (e^\prime,r),(L_s^6,l),1&\ra (L_s^7,i^\prime,1)\\
    (L_s^7,u),(L^{\prime},d),1&\ra (L_s^t,L^t,0)\\
    (x^t,r),(i^\prime,l),1&\ra (e^\prime,x^{t^\prime},1), x\in\{L,L_s\}\\
    (x^{t^\prime},r),(i^\prime,l),1&\ra (i^\prime,x^{t^\prime},1), x\in\{L,L_s\}\\
    (x^{t^\prime},r),(e^\prime,l),1&\ra (x^{t^\dprime},e,1), x\in\{L,L_s\}\\
    (i^\prime,r),(x^{t^\dprime},l),1&\ra (x^{t^\dprime},i,1), x\in\{L,L_s\}\\
    (e^\prime,r),(L_s^{t^\dprime},l),1&\ra (L_s,i,1)\\
    (e^\prime,r),(L^{t^\dprime},l),1&\ra (L_{start},i,1)                                    
    \phantom{\hspace{10cm}}
    \end{align*}
  \end{algorithmic}
\end{algorithm}

We additionally show that, in principle, the lines do not need a leader in order to successfully self-replicate. We give such a protocol which is ``more parallel'' and has a much more concise description than the previous one. We  now assume that one line has $e$ on both of its endpoints and $i$ on the internal nodes, and every (free) node is in state $q_0$. The code is presented in Protocol \ref{prot:no-leader-lr}. The protocol works as follows. Free nodes are attached below that nodes of the original line. When a node is attached below an internal node $i$ both become $i_1$ and when a node is attached below an endpoint $e$, both become $e_1$. Moreover, adjacent nodes of the replica connect to each other and every such connection increases their index. In fact, their index counts their degree. An internal node of the replica can detach from the original line only when it has degree 3, that is when, apart from its vertical connection, it has also already become connected to both a left and a right neighbor on the replica. On the other hand, an endpoint detaches when it has a single internal neighbor. It follows that the replica can only detach when its length (counted in number of horizontal active connections) is equal to that of the original line. To see this, assume that a shorter line detaches at some point. Clearly, such a line must have at least one endpoint that corresponds to an internal node $i_j$ of the replica. But this node is an endpoint of the shorter line, so its degree is less than 3, i.e. $j<3$, and we conclude that it cannot have detached. 

\floatname{algorithm}{Protocol}
\renewcommand{\algorithmiccomment}[1]{// #1}
\begin{algorithm}[!h]
  \caption{\emph{No-Leader-Line-Replication}}\label{prot:no-leader-lr}
  \begin{algorithmic}
    \medskip
    \State $Q=\{q_0,e,e_1,i,i_1,i_2,i_3\}$
    \State $\delta$: 
    \begin{align*}
    (i,d),(q_0,u),0&\ra (i_1,i_1,1)\\
    (e,d),(q_0,u),0&\ra (e_1,e_1,1)\\
    (i_j,r),(i_k,l),0&\ra (i_{j+1},i_{k+1},1) \text{ for all } j,k\in\{1,2\}\\
    (i_1,r),(e_1,l),0&\ra (i_2,e_2,1)\\
    (i_2,r),(e_1,l),0&\ra (i_3,e_2,1)\\    
    (e_1,r),(i_1,l),0&\ra (e_2,i_2,1)\\
    (e_1,r),(i_2,l),0&\ra (e_2,i_3,1)\\
    (i_3,u),(i_1,d),1&\ra (i,i,0)\\
    (e_2,u),(e_1,d),1&\ra (e,e,0)                                                   
    \phantom{\hspace{10cm}}
    \end{align*}
  \end{algorithmic}
\end{algorithm}

\begin{lemma} \label{lem:square}
There is a protocol (described above) that when executed on $n$ nodes (for all $n$ with integer $\sqrt{n}$) w.h.p. constructs a $\sqrt{n}\times \sqrt{n}$ square and terminates.
\end{lemma}
\begin{proof}
From Lemma \ref{lem:counting-line}, when the leader in Counting-on-a-Line protocol terminates, w.h.p. it has formed an active line of length $\log n$ containing $n$ in binary in the $r_0$ components of the nodes of the line. Then the leader computes $\sqrt{n}$ on its line and expands its line to make its length $\sqrt{n}$. Next the leader creates the seed replica by executing the routine described in Protocol \ref{prot:line-replication}. The seed replica keeps creating new self-replicating replicas. All these replications are performed by a routine essentially equivalent to Protocol \ref{prot:line-replication}. Every replica is a line of length $\sqrt{n}$ and will be eventually attached to the square-segment to form another row of the square. First observe that the seed may only be attached to the square, when the square has already obtained $\sqrt{n}-1$ rows. This implies that replications do not cease before the square has been successfully constructed. Additionally, any non-seed replica $r$ can be attached to the square-segment (whenever the $l$ leader is in the state of waiting for new attachments) independently of whether $r$ is in the middle of an incomplete replication. The reason is that attachment occurs via the up ports of $r$ while replication takes place via the down ports of $r$. If this occurs, then the nodes of the incomplete replication are simply released as free nodes. So, assume that there are $k$ nodes that are either free or part of an incomplete replication. We only have to prove that as long as $k\geq\sqrt{n}$ then eventually another replica has to be formed. If not, then for an infinite number of steps it holds that $k\geq\sqrt{n}$. Moreover, every non-seed replica in a finite number of steps becomes attached to the square-segment and releases any nodes of an incomplete replication. Thus, in a finite number of steps, every one of the $k\geq\sqrt{n}$ nodes is either free or part of an incomplete replication of the seed. Clearly, given that the seed does not cease self-replication and given that there are enough nodes to fill the $\sqrt{n}$ replication positions of the seed, in a finite number of steps (due to fairness) all these positions should have been filled and a replica should have been created. Thus, the assumption that no further replication occurs violates the fairness condition.
\qed
\end{proof}

\subsection{Simulating a TM}
\label{subsec:simulation}

We now assume as given (from the discussion of the previous section) a $\sqrt{n}\times\sqrt{n}$ square with a unique leader $L$ at the bottom left corner. However, keep in mind that, in principle, the simulation described here can begin before the construction of the $\sqrt{n}\times\sqrt{n}$ square is complete. The only difference in this case, is that the two processes are executed in parallel and if at some point the TM needs more space, it has to wait until it becomes available. The square may be viewed as a TM-tape of length $n$ traversed by the leader in a ``zig-zag'' fashion, first moving to the right until the bottom right corner is encountered, then one step up, then to the left until the node above the bottom left corner is encountered, then one step up again, then right, and so on. To simplify this process, we may assume that a preprocessing has marked appropriately the turning points (see Figure \ref{fig:shape2}). The tape will be used to simulate a TM $M$ of the form described in the Section \ref{sec:model}. The $n$ pixels of the square are numbered according to the above zig-zag process beginning from the bottom left node, each node corresponding to one pixel. The space available to the TM is exponential in the binary representation of the input $(i,n)$ (or $(i,\sqrt{n})$), because $i\leq n-1$ and therefore the length of its binary representation $|i|=O(\log n)$, thus $|(i,n)|=O(\log n)$, but the available space is $\Theta(n)=\Theta(2^{\log n})=\Omega(2^{|(i,n)|})$ (still it is linear in the size of the whole shape to be constructed).

The protocol invokes $n$ distinct simulations of $M$, one for each of the pixels $i\in\{0,1,\ldots,n-1\}$ beginning from $i=0$ and every time incrementing $i$ by one. The leader maintains the current value of $i$ in binary, in a pixel-counter $pixel$ stored in the $O(\log n)$ leftmost cells of the tape. \footnote{When we refer to the \emph{tape}, we mean the line produced by traversing the square in a zig-zag way beginning from the bottom-left node, as described above. So the ``leftmost'', here, corresponds to the leftmost nodes of the line, e.g. the left part of the bottom row of the square, and should not be confused with the nodes on the leftmost column of the square.} Recall that the leader knows $n$ from the procedures of the previous sections. So, we may assume that the tape also holds in advance $n$ and $\sqrt{n}$ in binary (again in the leftmost cells). Initially $pixel=0$ and the leader marks the 0th node, that is the bottom left corner of the square. Then it simulates $M$ on input $(pixel,\sqrt{n})$. When $M$ decides, if its decision is \emph{accept}, the leader marks the node corresponding to $pixel$ as \emph{on}, otherwise it marks it as \emph{off}. Then the leader increments $pixel$ by one, marks the node corresponding to the new value of $pixel$ (which is the next node on the tape), clears the tape from residues of the previous simulation, invokes another simulation of $M$ on the new value of $pixel$, and marks the corresponding node as \emph{on} or \emph{off} according to $M$'s decision. The process stops when $pixel=n$, in which case no further simulation is executed. Alternatively, the leader can detect termination by exploiting the fact that the last pixel to be examined is the one corresponding to the upper left or right corner of the square (depending on whether $\sqrt{n}$ is even or odd), which can be detected.

When the above procedure ends, the leader starts walking the tape in the opposite direction until it reaches the bottom left corner. In the way, it passes a \emph{release} signal to every node it goes through. A node enters the release phase exactly when the leader departs from that node, apart from the bottom left corner which enters the release phase when the leader arrives. When two nodes that are both in the release phase interact, if at least one of them is \emph{off} and their connection is active, they deactivate the connection. Clearly, the only nodes that will remain connected in the solution are the \emph{on} nodes forming the desired connected 2-dimensional shape that $M$ computes. If we additionally require the leader to know when all deactivations have completed and terminate, then we can either (i) have the leader deactivate them itself while moving backwards, also ensuring that it does not remain on a node that will be released, or (ii) have the leader repeatedly explore the final connected shape until it detects that all potential deactivations have occurred.

\begin{figure}[!hbtp]
   \centering{
        \subfigure[]{
        \includegraphics[width=0.35\textwidth]{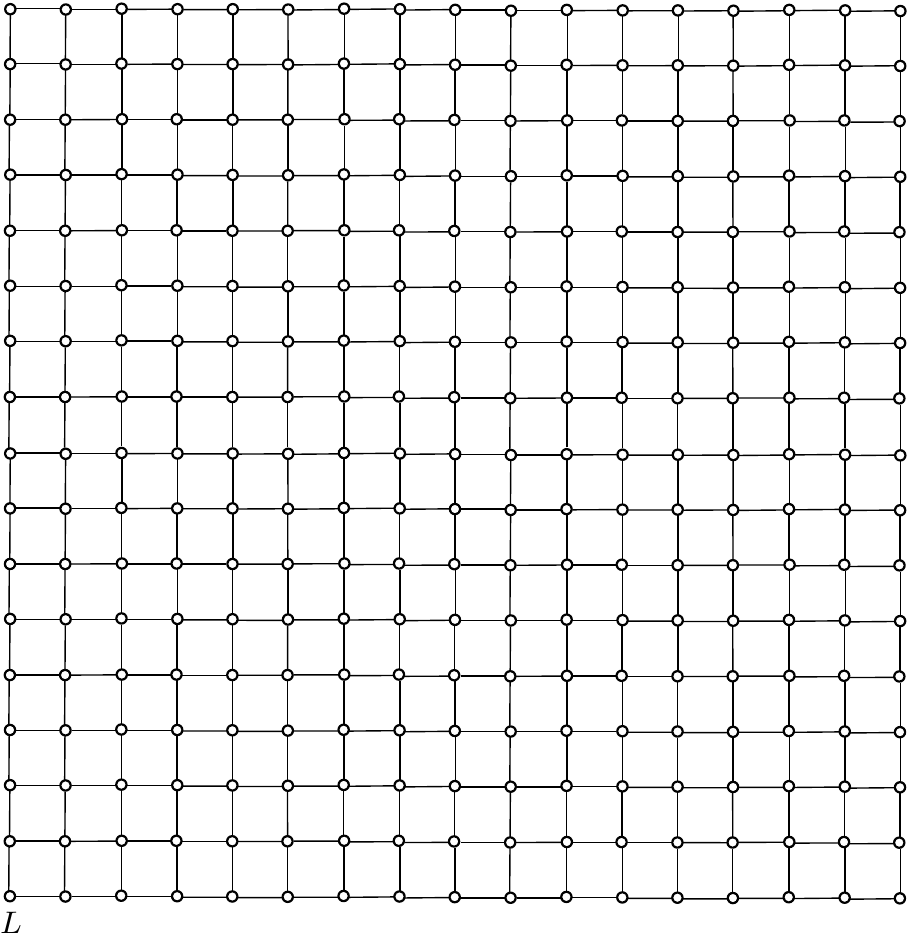}
        \label{fig:shape1}}
	\hspace{1cm}
        \subfigure[]{
        \includegraphics[width=0.35\textwidth]{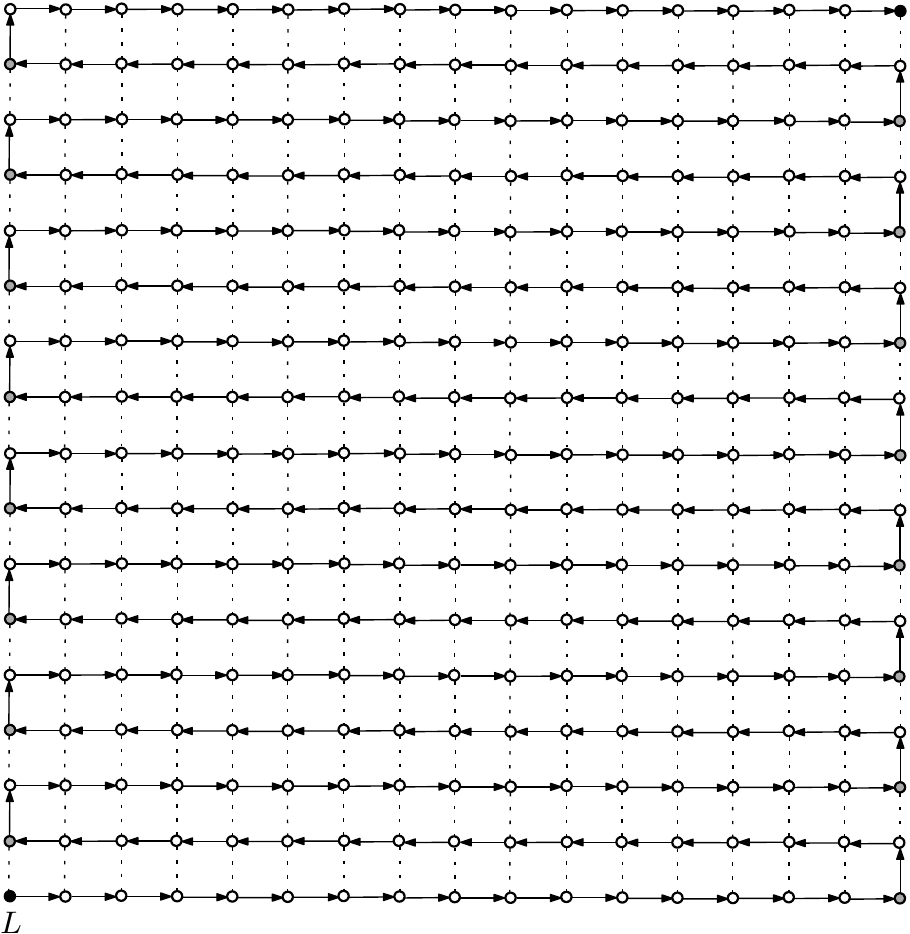}
        \label{fig:shape2}}
	\hspace{1cm}
        \subfigure[]{
        \includegraphics[width=0.35\textwidth]{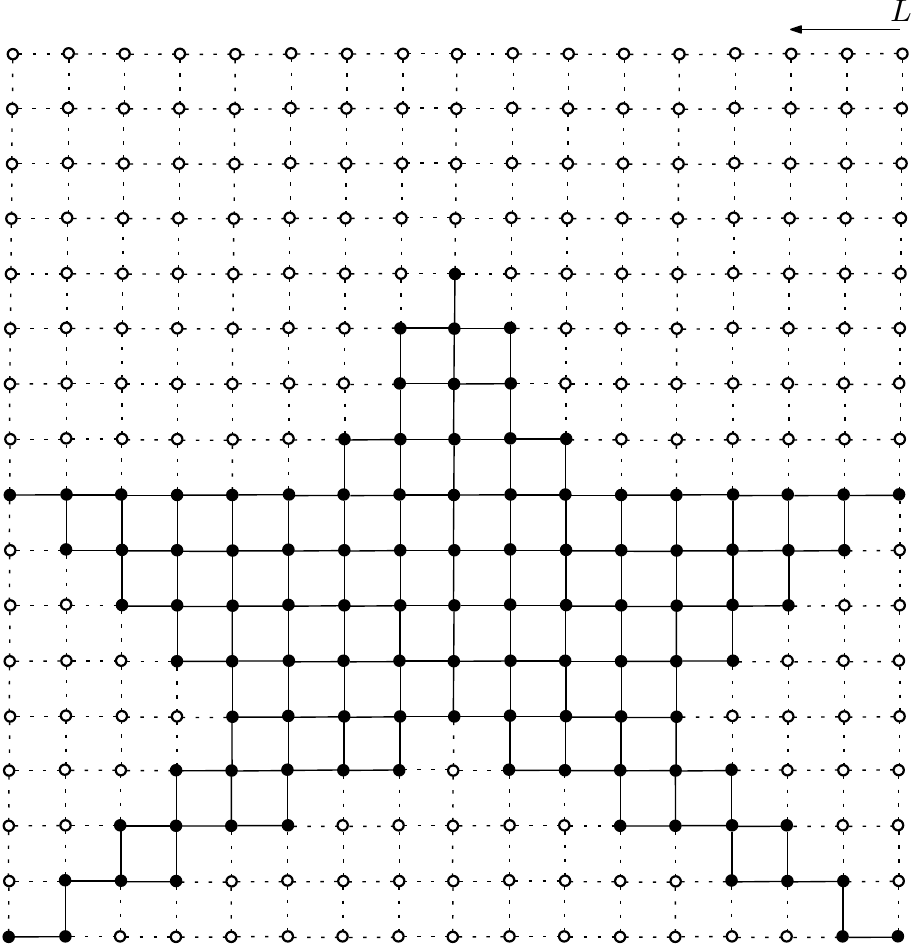}
        \label{fig:shape3}}
	\hspace{1cm}
        \subfigure[]{
        \includegraphics[width=0.35\textwidth]{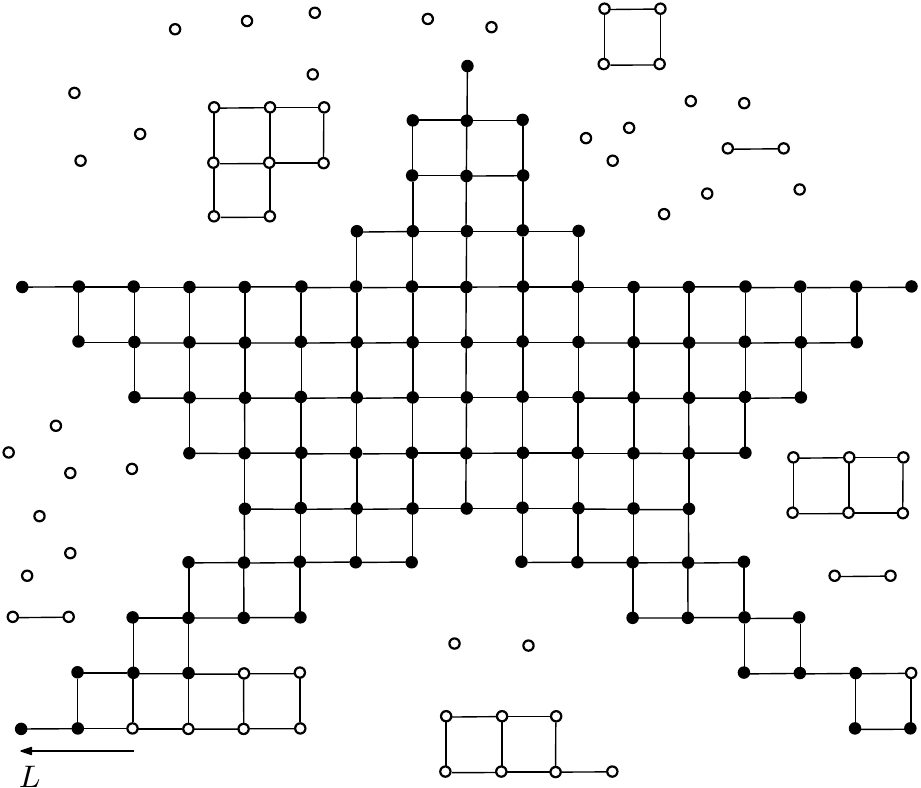}
        \label{fig:shape4}}        
        }
   \caption{(a) The $\sqrt{n}\times\sqrt{n}$ square has just been constructed. (b) The virtual tape on the square. The arrows show the direction in which the tape is traversed from left to right (opposite arrows for the opposite direction are not shown). The two endpoints of the tape are marked as black here and the turning points are marked as gray. These facilitate the leader to detect and choose the right action, e.g. turn left twice (equivalently, follow the up port and then the left port) when it arrives at the bottom right corner and wishes to continue on the second row. The indices of the pixels that the procedure assumes, follow the order of the tape, that is the first position of the tape corresponds to pixel 0, the second to pixel 1,..., the last position of the tape to pixel $n-1$. (c) The shape, which looks like a star, has been formed on the square. It consists of the pixels that the TM $M$ decided to be \emph{on}, which are colored black here. All other white pixels are the \emph{off} pixels. The simulations have completed and the leader has just reached the upper right corner and now it starts releasing the shape. To improve visibility, the edges that will eventually be deactivated appear as dotted here. (d) Releasing is almost complete. The leader has reached the bottom left corner and has updated all nodes to the release phase. Any connection involving at least one \emph{off} node (i.e. a white one) will be eventually deactivated.} \label{fig:shape}
\end{figure}

The following theorem states the lower bound implied by the construction described in this section.

\begin{theorem}
Let $\cl=(S_1,S_2,\ldots)$ be a connected 2D shape language, such that $\cl$ is TM-computable in space $d^2$. Then there is a protocol (described above) that w.h.p. constructs $\cl$. In particular, for all $d\geq 1$, whenever the protocol is executed on a population of size $n=d^2$, w.h.p. it constructs $S_d$ and terminates. In the worst case, when $G_d$ (that is, the shape of $S_d$) is a line of length $d$, the waste is $(d-1)d=O(d^2)=O(n)$.
\end{theorem}
\begin{proof}
We have to show that for every $n=d^2$, when the protocol is executed on $d^2$ nodes constructs $G_d$. From Lemma \ref{lem:square}, we have a subroutine that terminates having w.h.p. constructed a $d\times d$ square with a unique leader on the bottom left node. Next, the leader can easily organize the square into a tape of length $d^2$ that has $d$ stored in binary in its leftmost cells. Moreover, $\cl$ is computable, so, by Definition \ref{def:computable}, there is a TM $M$ that when executed on the pixels of a $d\times d$ square constructs $S_d$. The protocol simulates $M$ on the pixels of such a $d\times d$ square thus the result is $S_d$, which is an \emph{on}/\emph{off} labeled $d\times d$ square the \emph{on} pixels of which form $G_d$. To perform the simulation, the protocol just feeds $M$ with $(i,d)=(0,d),(1,d),\ldots, (d^2-1,d)$, one at a time, simulates $M$ on input $(i,d)$ in space $\Theta(d^2)$, marks the corresponding pixel as \emph{on} or \emph{off} according to $M$'s decision, and moves on to the next input. When $i=d^2$, the square contains $G_d$ and the leader releases $G_d$ by one of the terminating approaches described above and terminates. Observe that, given the guarantees of Lemma \ref{lem:square}, the procedure described here is always correct. So, the probability of failure of the whole protocol is just the probability of failure of the initial counting subroutine, thus the protocol succeeds w.h.p.. Finally, the waste is always equal to the number of pixels of the $d\times d$ square that are not part of $G_d$. Observe now that the waste can never be more than $(d-1)d$, because if it was at least $(d-1)d+1=d^2-d+1$, then the size of $G_d$ (i.e. the useful space) would be at most $d^2-(d^2-d+1)=d-1$. But then, connectivity of $G_d$ implies that $max\_dim_{G_d}\leq d-1$, which contradicts the assumption that $max\_dim_{G_d}=d$. Thus, the worst possible waste is indeed $(d-1)d=O(d^2)=O(n)$. Notice that here the waste of the protocol is equal to the waste of the simulated TM: the protocol just provides the maximum square that fits in the population and the TM determines which nodes will be part of the final shape and which will be thrown away as waste.
\qed
\end{proof}

\begin{remark}
It is worth mentioning that if the system designer knew $n$ in advance, then he/she could preprogram the nodes to simulate a TM that constructs a specific shape of size $n$, for example the TM corresponding to the Kolmogorov complexity of the shape (which is in turn the Kolmogorov complexity of the desired binary pixel sequence $(s_0,s_1,\ldots,s_{n-1})$). However, in this work we consider systems in which $n$ \emph{is not known in advance}, so the natural approach is to preprogram the nodes with a TM that can work for all $n$. The protocol must first compute $n$ (w.h.p.) and then simulate the TM on input $n$ to construct a shape of the appropriate size. For example, it could be a TM constructing a star, as in Figure \ref{fig:shape3}, such that the size of the star grows as $n$ grows.   
\end{remark}

\begin{remark}
The above results can be immediately modified to refer to \emph{patterns} instead of shapes. In fact, observe that the $\sqrt{n}\times\sqrt{n}$ square that has been labeled by \emph{off} and \emph{on} by the TM is already such a (computed) 0/1 pattern. The generic idea to extend this is to keep the same constructor as above and simulate TMs that for every pixel output a color from a set of colors $\cc$. Then the resulting square with its nodes labeled from $\cc$ is the desired computed pattern and no releasing is required in this case.
\end{remark}

\subsection{Parallelizing the Simulations}
\label{subsec:parsim}

We now present two approaches for parallelizing the simulations, instead of executing them sequentially one after the other.

\subsubsection{Approach 1}

The first approach uses the 3D model (that is, the one with 6 ports) to construct a 2D shape as before. Assume that the shape $G$ has $\max\_dim_G=d$ ($d$ could be a function of $n$; e.g. it was $O(\sqrt{n})$ in the previous section). We again construct a $d\times d$ square as before. Assume also that there is a TM $M$ deciding the status of each of the pixels in space $k$ (again could be a function of $n$, e.g. $\log n$ or $\sqrt[3]{n}$). We assume that both $k$ and $d$ are computable in space $O(n)$, so that in the worst case we can create a spanning line, compute them on it by simulating a TM, and create seeds of the appropriate lengths as before. Assume also that we have a population of size $k\cdot d^2$.

When counting of $n$ terminates, the leader computes two different seeds, one of length $d$ and another of length $k-1$ (this is in contrast to the unique seed of Section \ref{subsec:repl-square}). Then it first activates the seed of length $d$ and keeps the other seed in a sleeping state. As before, it constructs a $d\times d$ square using, say, only dimensions $x$ and $y$. When this construction completes, the leader wakes up the seed of length $k-1$ to organize all the remaining nodes into lines of length $k-1$ in the $z$ dimension. Each of these lines will be attached ``below'' \footnote{It is not actually below; it is in the positive part of the $z$ dimension, but we use ``below'' to be consistent with the more intuitive illustration of Figure \ref{fig:3d-par-simulations}.} one of the pixels of the square. The pixel and its line form together the required TM-tape of length $k$ corresponding to that pixel (see Figure \ref{fig:3d-par-simulations}). So, when this process ends, we have a memory of length $k$ attached to each pixel of the square. For a simple example, when $d=k=\sqrt[3]{n}$, then the construction so far will ``look like'' a $\sqrt[3]{n}\times \sqrt[3]{n}\times \sqrt[3]{n}$ cube (actually, if for physical reasons we want to have a more rigid structure, we can also activate all edges between the lines; in this case the construction of the previous example \emph{would be} a $\sqrt[3]{n}\times \sqrt[3]{n}\times \sqrt[3]{n}$ cube). Then the leader initializes each memory to $(i,d)$, where $i$ is the index of its pixel, the indices being counted by their distance from the ``bottom left'' corner of the square as before, and then informs each pixel-head to start simulating $M$ on its tape. In this way, we have $d^2$ simulations being executed in parallel each on its own tape of length $k$. When all simulations have completed, the leader may first release all memories to keep only the $d\times d$ square and then apply the releasing process described in the previous section to release the \emph{off} pixels of the square and isolate the connected shape consisting of the pixels that are \emph{on}.

\begin{figure}[!hbtp]
\centering{
\includegraphics[width=0.55\textwidth]{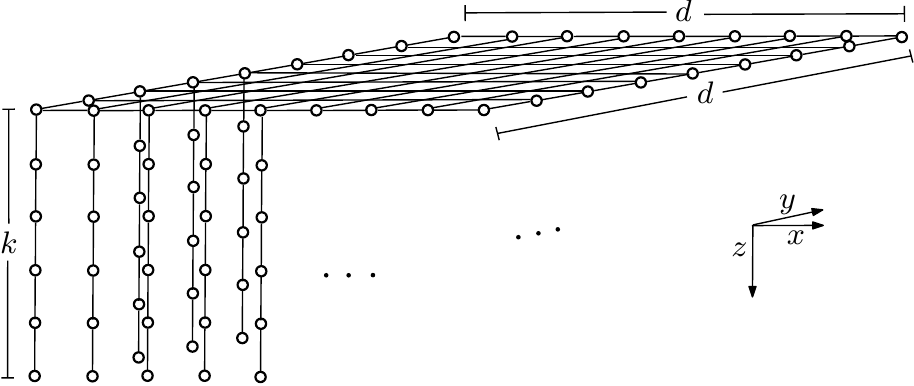}
}
\caption{The constructed $d\times d$ square lies in dimensions $x$ and $y$. We can think as its ``bottom left'' corner, its leftmost node in the figure. Every internal intersection point of the square is also a node, but we have not drawn these nodes here to improve visibility. ``Below'' it, in dimension $z$, are the $d^2$ lines of length $k$ each. The protocol executes a distinct simulation of the TM on each of these lines. In particular, on the line attached to pixel $i$, for all $0\leq i\leq d^2-1$, the protocol simulates the TM on input $(i,d)$.} \label{fig:3d-par-simulations}
\end{figure}

\begin{theorem}
Let $\cl=(S_1,S_2,\ldots)$ be a TM-computable connected 2D shape language, such that $S_d$ is computable in space $k=f(d)$ and $k$ is computable in space $O(k\cdot d^2)$. Then there is a protocol (described above) that w.h.p. constructs $\cl$. In particular, for all $d\geq 1$, whenever the protocol is executed on a population of size $n=k\cdot d^2$, w.h.p. it constructs $S_d$ and terminates, by executing $d^2$ simulations in parallel each with space $O(k)$. In the worst case, when $G_d$ is a line of length $d$, the waste is $(d-1)d+(k-1)d^2=O(k\cdot d^2)$.
\end{theorem}

\subsubsection{Approach 2}

We now show how to achieve a similar parallelism while avoiding the use of a third dimension. Now the unique leader that knows $n$, instead of constructing a square, constructs a spanning line of length $d^2$, say in the $x$ dimension. This line corresponds to a linear expansion of the pixels of the $d\times d$ square of the previous construction. Moreover, the leader creates a seed of length $k-1$ as before, to partition the rest of the nodes into lines of length $k-1$, this time in the $y$ dimension. Each such line will be attached below one of the nodes of the $x$-line. As before, when all $y$-lines have been attached, the leader initializes their memories with $(i,d)$, where $i$ is the index of the corresponding pixel (the index of each pixel is now its distance from the left endpoint of the $x$-line, beginning from 0 and ending at $d^2-1$). Then all simulations of $M$ are executed in parallel and eventually each one of them sets its $x$-pixel to either \emph{on} or \emph{off}. When all simulations have ended, the leader releases the auxiliary memories (i.e. the $y$-lines) and then partitions the $x$-line into consecutive segments of length $d$ by placing appropriate marks on the boundary nodes (see Figure \ref{fig:linear-par-sim}). Each segment corresponds to a row of the $d\times d$ square to be constructed. In particular, segment $i\geq 1$ counting from left corresponds to row $i$ (rows being counted bottom-up). Observe that, in the way the pixels have been indexed, segment $2$ should match with its upper side the upper side of segment $1$ (that is segment 2 should rotate $180\degree$), segment $3$ should match with its lower side the lower side of segment $2$, and so on. In general, if $i$ is even, segment $i$ should match with its upper side to the upper side of segment $i-1$ and, if $i$ is odd, segment $i$ should match with its lower side the lower side of segment $i-1$. The leader marks appropriately the nodes of each segment to make them aware of the orientation that they should have in the square. Moreover, it assigns a unique key-marking to each segment so that segment $i$ can easily and locally detect segment $i-1$. In particular, if $i$ is odd, it marks nodes $i$ and $i-1$ of the segment counting from left to right (for segment 1 it only marks the leftmost node), and, if $i$ is even, it marks nodes $i$ and $i-1$ of the segment counting from right to left. In this manner, given that segments respect the correct orientation and provided that attachment is only performed when their endpoints match, every segment $i$ uniquely matches to segment $i-1$ because the first mark of $i$ is uniquely aligned with the second mark of $i-1$ (see Figure \ref{fig:par-sim-assemblying}). Then the leader releases all segments, one after the other, and it remains on the last segment. The segments are free to move in the solution until they meet their counterpart, and when this occurs the two segments bind together. Eventually, the $d\times d$ square is constructed and every pixel is in the correct position (the position corresponding to its index counting in a zig zag fashion as in the previous sections). The leader periodically walks on its component to detect when it has become equal to the desired square. When this occurs, it initiates as before the releasing phase to isolate the final connected shape consisting of the \emph{on} pixels.   

\begin{figure}[!hbtp]
   \centering{
        \subfigure[]{
        \includegraphics[width=0.73\textwidth]{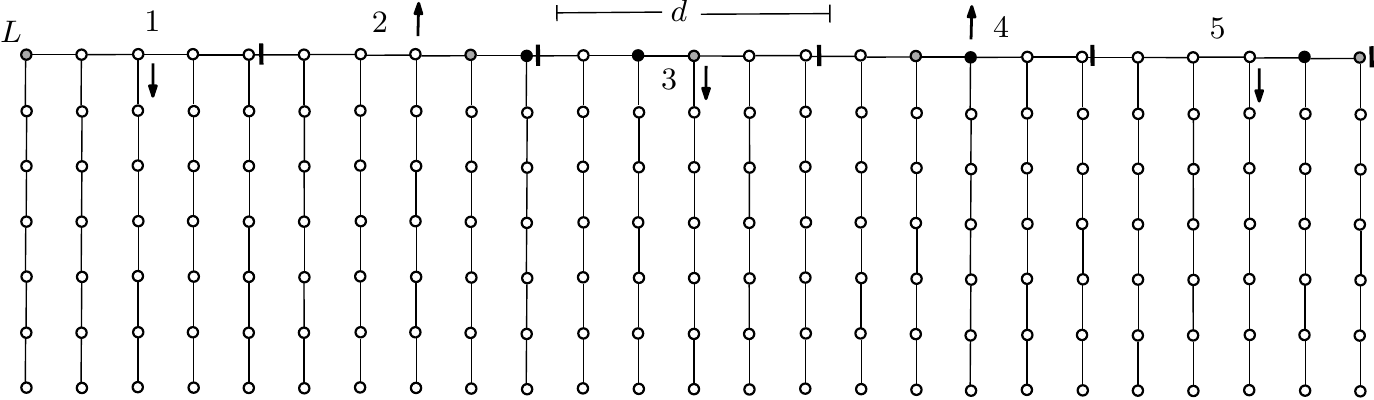}
        \label{fig:linear-par-sim}}
	\hspace{1cm}
        \subfigure[]{
        \includegraphics[width=0.17\textwidth]{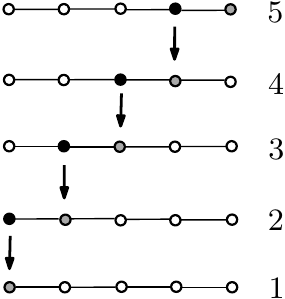}
        \label{fig:par-sim-assemblying}}  
        }
   \caption{(a) As in Figure \ref{fig:3d-par-simulations}, $d^2$ lines of length $k-1$ each, are pendent below the $d^2$ pixels. The difference now is that the pixels have been arranged linearly in dimension $x$. So, the whole construction is now 2-dimensional. The pixels have been partitioned into equal segments of length $d$ each (see the black vertical delimiters). The numbers represent the indices of the segments counted from left to right. The arrows leaving above or below the segments, indicate which side of the segment should look ``downwards'' in the square that will be constructed. For example, segment 1 can remain as it is, while segment 2 has to be rotated so that its upper side attaches to the upper side of segment 1. Every segment has been marked by a black and a gray node placed at an appropriate position. (b) The segments have been released in the solution, and now they have to gather together in order to form the square. Each segment knows the correct orientation, i.e. whether it should use its up or down ports, and also it can detect its predecessor row by exploiting the marking. In particular, it attaches to a row if its black mark is above the gray mark of the other row when their orientation is correct and their endpoints are totally aligned.} \label{fig:linear-par}
\end{figure}

\begin{remark}
In all the above constructions the unique leader assumption can be dropped in the price of sacrificing termination. In this case, the constructions become stabilizing by the reinitialization technique, as e.g. in \cite{MS14}, but should be carefully rewritten.
\end{remark}

\section{Replicating Arbitrary 2D Shapes}
\label{sec:replication}

In this section we consider the problem of replicating a given 2-dimensional shape $G$ without using a third dimension. $G$ is a connected shape with its nodes labeled \emph{on} and has a unique leader on one of its nodes. The leader does not know $G$ nor its size. All remaining nodes are \emph{off} and free in the solution and we assume that is in any case a sufficient number of them (their actual number depends on $G$ and the replication approach).

\subsection{Approach 1}

First a \emph{squaring} phase is executed during which the shape $G$ is being surrounded with additional nodes in order to become included in the smallest rectangle $R_G$ containing it (observe that ``squaring'' is actually a misnomer). As we discuss below, the leader need not control the squaring phase as it can be performed in parallel by all nodes of the square by local tests and expansions. However, the leader is the one that will detect the successful completion of the squaring phase. To achieve this, it only has to move periodically around the connected shape until it detects that it has become a rectangle. When this occurs, the squaring phase ends and the leader initiates the \emph{shifting} phase. W.l.o.g. let shifting be performed in terms of columns being shifted in the $x$ dimension and let the leader begin from the leftmost column. The leader first copies the configuration of the leftmost column (i.e. the state, \emph{on} or \emph{off}, of each node of the column) to separate components in the states of the nodes of the second column. The original status of each node, that is whether it was initially \emph{on} or \emph{off}, is always maintained in a component of its state. The additional components store the replica that is being shifted to the right. In general, the leader copies, one after the other, the configuration of column $i$ to column $i+1$, for every column $i$ of the rectangle $R$ apart from the rightmost one. When it reaches the rightmost column, which is the last column to be copied, as there is no further column to the right, the leader attaches first free nodes to create a new column and then performs copying as before. This completes the first shifting round. In all subsequent shifting rounds, the leader performs shifting always beginning from the rightmost column of the replica. Again it first introduces a new column to the right in order to shift the rightmost column of the replica and then starts, one after the other, to shift the remaining columns to the right. When it completes another shifting round, it moves again to the rightmost column of the replica in order to perform the next shifting round in precisely the same way. When a shifting rounds ends at the rightmost column of the original rectangle, then the leader knows that the whole shape has been totally shifted to an identical rectangle to the right and it stops the shifting phase. Then it releases the two rectangles by deactivating the connections between the rightmost column of the original and the leftmost column of the replica. After this, one leader remains on the original rectangle and the other on the replica. Both execute a de-squaring phase to release the dummy nodes that were used to form the rectangles and isolate the two identical shapes.\\

\noindent\textbf{Squaring.} We give some more details on the squaring process. Recall that we have a given, called the \emph{original}, connected shape $G$ that we want to enclose in the minimal rectangle $R_G$ containing it. We claim that squaring can be performed by local detections and actions without the need of a unique leader in the shape. 

\begin{figure}[!hbtp]
   \centering{
        \subfigure[]{
        \includegraphics[width=0.2\textwidth]{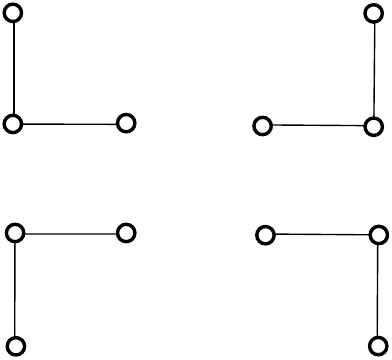}
        \label{fig:det-patterns1}}
	\hspace{1cm}
        \subfigure[]{
        \includegraphics[width=0.4\textwidth]{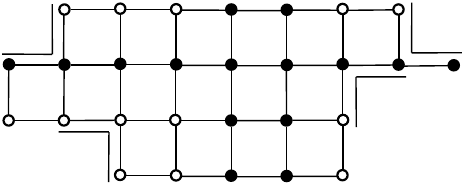}
        \label{fig:det-patterns2}}  
        }
   \caption{(a) If the connected shape is not yet a rectangle then at least one of these locally detectable shapes must exist. (b) An example of an incomplete rectangle, in which all four detection shapes appear. Black nodes and the connections between them constitute the original shape $G$. White nodes and the remaining connections have been introduced by the so far execution of the squaring process.} \label{fig:det-patterns}
\end{figure}

\begin{proposition}
At least one of the shapes of Figure \ref{fig:det-patterns1} exists in a connected shape $G$ iff $G$ is not a rectangle. Such shapes can be used to locally detect that $G$ is not a rectangle yet.
\end{proposition}
\begin{proof}
Clearly, as long as $G$ is still a proper subgraph of $R_G$, at least one edge or node (or both) is present in $R_G$ but not in $G$. Observe that if an edge is missing but its endpoint nodes are both present and connected by active edges to the rest of the shape, then it is trivial to detect the absence of the edge locally: each of the two endpoints knows that it is connected to the shape, thus the two nodes just have to activate the edge joining them when they interact over it. So, we can w.l.o.g. assume that all possible edges are present between the nodes of $G$ and focus on the case that some node is missing. In fact, if some node is missing, then it must hold that at least two rows or at least two columns of $G$ have unequal lengths. Let us only consider the row case, as the other is symmetric. If there exist two rows in $G$ with unequal lengths then there must necessarily exist two \emph{consecutive} rows of $G$ with unequal lengths (otherwise, if all consecutive rows preserved the length then all rows should have the same length). As $G$ is connected, there must be at least one active vertical edge joining the two rows. Then we can begin from that edge and by walking either to the left or to the right, we must meet the first position at which the two rows differ. For example, if it was to the right, then node $u$ of one row has a right neighbor $u_r$, while the node $v$ above it (or below it, respectively) does not. Then it is trivial for the triple $(u, u_r, v)$ to locally detect that $v_r$ is missing and be ready to attach a free node to the required position when such a node arrives. Inversely, if $G$ is a rectangle then none of the above locally detectable shapes can exist.
\qed
\end{proof}

So, we can have all outer nodes of the segment of $R_G$ that has been constructed so far, to locally handle the squaring process without the help of the leader. The leader is only required up to this point to detect termination of the squaring phase. In particular, it suffices for the leader to detect that its component has become a rectangle, because the above process guarantees that when this occurs the constructed rectangle must be equal to $R_G$. TO achieve this the leader can, for example, begin from the leftmost column that it knows and attempt to traverse a rectangle in a zig-zag way, e.g. up the first column, then right, then down the second column, and so on. If it manages to complete such a traversal without encountering any recesses or overhangs. Finally, it is worth mentioning that, in order to execute correctly, the above replication protocol requires a population of size at least $2|V(R_G)|$ and its waste is $2(|V(R_G)|-|V(G)|)$.

\subsection{Approach 2}

Again the original shape is first squared. Then each column is assigned a unique matching identifier as in the previous section, so that column $i$ matches uniquely to column $i-1$. Then the rightmost column is replicated, by attaching free nodes to its right. Replication includes also the unique matching key. When replication completes, first the replica-column is released and then the original column, so that both move freely in the solution. If desired, we can have replica-columns (or just their keys) to use different states than original columns so that no two columns of different kinds ever become connected. The process continues with the other columns, each time replicating the rightmost one and then releasing both the replica and the original column. Eventually, all columns will be replicated and released. Moreover, as original (replica) column $i$ uniquely matches to original (replica) column $i-1$, eventually both the original and the replica rectangles are correctly assembled. Finally, a de-squaring phase is executed as before on both rectangles, to release the dummy nodes that were used for forming the rectangles and isolate the two identical shapes.

\section{Conclusions and Further Research}
\label{sec:conclusions}

There are several interesting open problems related to the findings of this work. A possible refinement of the model could be a distinction between the \emph{speed of the scheduler} and the \emph{internal operation speed of a component}. For example, a connected component will operate in synchronous rounds, where in each round a node observes its neighborhood and its own state and updates its state based on what it sees. Nodes can of course update also the state of their local connections and we may assume that a connection is formed/dropped if both nodes agree on the decision (another possibility is to allow a link change state if at least one of the nodes say so). This distinction between two different ``times'', though ignored so far in the literature, is very natural because a connected component should operate at a different speed than it takes for the scheduler to bring two nodes (e.g. of different components, or an isolated node and a node of some component) into contact. 

It would be also interesting to consider for the first time a \emph{hybrid model} combining active mobility controlled by the protocol and passive mobility controlled by the environment. For example it could be a combination of the Nubot model and the model presented in this work. Another very intriguing problem is to give a proof, or strong experimental evidence, of Conjecture \ref{conj:count-impossibility}. If true, it would imply that there is no analogue of Theorem \ref{the:count-half} if all processes are identical. A possibility left open then would be to achieve high probability counting with $f(n)$ leaders. There is also work to be done w.r.t. analyzing the running times of our protocols and our generic constructors and proposing more efficient solutions. Also it is not yet clear whether the protocol of Section \ref{subsec:leader-counting} is the fastest possible nor that its success probability or the upper bound on $n$ that it guarantees cannot be improved; a proof would be useful. Moreover, it is not obvious what is the class of shapes and patterns that the TMs considered here compute. Of course, it was sufficient as a first step to draw the analogy to such TMs because it helped us establish that our model is quite powerful. However, still we would like to have a characterization that gives some more insight to the actual shapes and patterns that the model can construct.

It would be also important to develop models (e.g. variations of the one proposed here) that take other real physical considerations into account. In this work, we have restricted attention on some geometrical constraints. Other properties of interest could be weight, mass, strength of bonds, rigid and elastic structure, collisions, and the interplay of these with the interaction pattern and the protocol. Moreover, in real applications mere shape construction will not be sufficient. Typically, we will desire to output a shape/structure that \emph{optimizes some global property}, like energy and strength, or that achieves a \emph{desired behavior in the given physical environment}. The latter also indicates that the construction and the environment that the construction inhabits cannot be studied in isolation. Instead, the two will constantly affect each other, the optimal output will highly depend on the optimality that the environment allows and also the environment may highly and continuously affect the construction process. The capability of the environment to affect the construction process suggests many robustness issues. Imagine an environment that can at any given time break an active link with some (small) probability (a similar question was also posed to the author during his talk at PODC '14 by some attendee, which the author would like to thank). Under such a perpetual setback no construction can ever stabilize. However, we may still be able to have a construction that constantly exists in the population by evolving and self-replicating.

Finally, in the same spirit, it would be interesting to develop routines that can rapidly reconstruct broken parts. For example, imagine that a shape has stabilized but a part of it detaches, all the connections of the part become deactivated, and all its nodes become free. Can we detect and reconstruct the broken part efficiently (and without resetting the whole population and repeating the construction from the beginning)? What knowledge about the whole shape should the nodes have to be able to reconstruct missing parts of it?\\

\noindent \textbf{Acknowledgements.} The author would like to thank Dimitrios Amaxilatis and Marios Logaras for implementing (in Java) the probabilistic counting protocol of Section \ref{subsec:leader-counting} and experimentally verifying its correctness and also David Doty for a few fruitful discussions on the same protocol at the very early stages of this work.

\newpage

\newcommand{\etalchar}[1]{$^{#1}$}

\end{document}